\documentclass[11pt]{article}
\usepackage{amsmath}
\usepackage{amsthm}
\usepackage{thmtools, thm-restate}
\usepackage{amssymb}
\usepackage{fullpage}
\usepackage{algorithmicx,algpseudocode}
\usepackage{algorithm}
\usepackage{xcolor}
\usepackage[english]{babel}
\usepackage{graphicx}
\graphicspath{{images/}{../images/}}
\usepackage{subfiles}
\usepackage{blindtext}
\usepackage{enumitem}
\usepackage[numbers]{natbib}

\newcommand{\ALG}{\text{ALG}}
\newcommand{\OPT}{\text{OPT}}
\newcommand{\Ex}{\mathbb{E}}
\newtheorem{theorem}{Theorem}
\newtheorem{lemma}{Lemma}
\newtheorem{definition}{Definition}
\newtheorem{prop}{Proposition}

\newtheorem{remark}{Remark}
\newcommand{\nameOfProblem}{submodular $k$-secretary problem}
\newcommand{\nameOfProblemSL}{submodular $k$-secretary problem with shortlists}
\newcommand{\streamingProblem}{submodular random order streaming problem}
\newcommand{\nameOfProblemSLplus}{submodular $(k+1)$-secretary problem with shortlists}
\newcommand{\nameSec}{secretary problem}
\newcommand{\nameSecSL}{secretary problem with shortlists}

\newcommand{\scomment}[1]{}
\newcommand{\ccomment}[1]{}
\newcommand{\mcomment}[1]{}
\newcommand{\toRemove}[1]{}
\newcommand{\etaMacro}{$c\frac{\log(1/\epsilon)}{\epsilon^2}  {\frac{1}{\epsilon^6} \log(1/\epsilon)  \choose {\frac{1}{\epsilon^4} \log(1/\epsilon)}}$~}
\usepackage{mathtools}

\DeclarePairedDelimiter\floor{\lfloor}{\rfloor}

\title{Submodular Secretary Problem with Shortlists}

\date{ }

\begin{document}

\author{Shipra Agrawal  \thanks{Columbia University, \texttt{sa3305@columbia.edu}. Research supported in part by Google Faculty Research Awards 2017 and Amazon Research Awards 2017.
}	
\and 	Mohammad Shadravan \thanks{Columbia University, \texttt{ms4961@columbia.edu}} 	
\and 	Cliff Stein \thanks{Columbia University, \texttt{cliff@ieor.columbia.edu}. Research supported in part by NSF grants CCF-1421161 and CCF-1714818.} }





\maketitle


\begin{abstract} 
In \nameOfProblem, 
the goal is to select $k$ items in a randomly ordered input so as to maximize 
the expected value of a given monotone submodular function on the set of selected items. 
In this paper, we introduce a relaxation of this problem, which we refer to as \nameOfProblemSL.
In the proposed problem setting, the algorithm is allowed to choose more than $k$ items as part of a shortlist. Then, after seeing the entire input, the algorithm can choose a subset of size $k$ from the bigger set of items in the shortlist. We are interested in understanding to what extent this relaxation can improve the achievable competitive ratio for the \nameOfProblem. In particular, using an $O(k)$ shortlist, can an online algorithm achieve a competitive ratio close to the best achievable offline approximation factor for this problem? 

We answer this question affirmatively by giving a polynomial time algorithm that achieves a $1-1/e-\epsilon-O(k^{-1})$ competitive ratio for any constant $\epsilon>0$, using a shortlist of size $\eta_\epsilon(k)=O(k)$. This is especially surprising considering that the best known competitive ratio (in polynomial time) for the \nameOfProblem~is $(1/e-O(k^{-1/2}))(1-1/e)$ \cite{kesselheim}. Further, for the special case of $m$-submodular functions, we demonstrate an algorithm that achieves $1-\epsilon$ competitive ratio for any constant $\epsilon>0$, using an $O(1)$ shortlist.

The proposed algorithm also has significant implications for another important problem of submodular function maximization under random order streaming model and $k$-cardinality constraint. We show that our algorithm can be implemented in the streaming setting using a memory buffer of size $\eta_\epsilon(k)=O(k)$ to achieve a $1-1/e-\epsilon-O(k^{-1})$ approximation. 
This substantially improves upon \cite{norouzi}, which achieved the previously best known approximation factor of $1/2 + 8\times 10^{-14}$ using $O(k\log k)$  memory.

\end{abstract}

\section{Introduction}
 

In the classic \textit{secretary problem}, 
$n$ items appear in random order.  We know $n$, but don't know the value of an item until it appears.   
Once an item arrives we have to irrevocably and immediately decide whether or not to select it. Only one item is allowed to be selected, and the objective is to  select the most valuable item, or perhaps to 
maximize the expected value of the selected item  
~\cite{Dynkin:SovMath:1963,ferguson1989solved, 10.2307/2985407}. 
It is well known that the optimal policy is to observe the first $n/e$ items without making any selection and then select the first item whose value is larger than the value of the best item in the first $n/e$ items~\cite{Dynkin:SovMath:1963}. This algorithm, given by ~\citet{Dynkin:SovMath:1963},
is asymptotically optimal, and hires the best secretary with probability at least $1/e$. Hence it
is also $1/e$-competitive for the expected value of the chosen item, and
it can be shown that no algorithm can beat $1/e$-competitive ratio in expectation.

Many variants and generalizations of the secretary problem have been studied in the literature, see e.g., \cite{Ajtai:2001,WILSON1991325,vanderbei1980optimal, wilson1991optimal, kleinberg, Babaioff:2008}.  
\cite{kleinberg, Babaioff:2008} introduced a 
{multiple choice secretary problem}, where the goal is to select $k$ items in a randomly ordered input so as to maximize  the {\it sum} of their values; and  \citet{kleinberg} gave an algorithm with an asymptotic competitive ratio of $1-O(1/\sqrt{k})$. Thus as $k\to \infty$, the competitive ratio approaches  1. Recent literature studied several generalizations of this setting to multidimensional knapsacks \cite{moser1997algorithm}, and proposed algorithms for which the expected online solution approaches the best offline solution as the knapsack sizes becomes large~(e.g., \cite{FHKMS10, devanur-hayes, AgrawalWY14}). 

In another variant of multiple-choice secretary problem, \citet{Bateni} and ~\citet{Gupta:2010} introduce 
the \textit{\nameOfProblem}. In this secretary problem, the algorithm again selects $k$ items, but 
the value of the selected items is given by a monotone submodular function 
$f$. The algorithm has a value oracle access to the function, i.e., for any given set $T$, an algorithm can query an oracle to find its value $f(T)$~\cite{oracle}.
The algorithm can select at most $k$ items $a_1 \cdots, a_k$, from a randomly ordered sequence of $n$ items. The goal is to maximize $f(\{a_1,\cdots, a_k\})$.
Currently, the best result for this setting is due to ~\citet{kesselheim}, who achieve a $1/e$-competitive ratio in exponential time, or $\frac{1}{e}(1-\frac{1}{e})$ in polynomial time. In this case, the offline problem is NP-hard and hard-to approximate beyond the factor of $1-1/e$ achieved by the greedy algorithm \cite{nemhauser1978best}. 
However, it is unclear if a competitive ratio of $1-1/e$ can be achieved by an online algorithm for the {\nameOfProblem} even when $k$ is large. 

\paragraph{Our model: \nameSecSL.}
In this paper, we consider a relaxation of the secretary problem 
where the algorithm is allowed to  select a {\it shortlist} of items that is larger than the number of items that ultimately need to be selected. That is, in a multiple-choice secretary problem with cardinality constraint $k$, the algorithm is allowed to choose more than $k$ items as part of a shortlist. Then, after seeing the entire input, the algorithm can choose a subset of size $k$ from the bigger set of items in the shortlist. 

This new model is motivated by some practical applications of secretary problems, such as hiring (or assignment problems),  where in some cases it may be possible to tentatively accept a larger number of candidates (or requests), while deferring the choice of the final $k$-selections to after all the candidates have been seen. Since there may be a penalty for declining candidates who were part of the shortlist, one would prefer that the shortlist is not much larger than $k$.

Another important motivation is theoretical: we wish to understand to what extent this relaxation of the secretary problem can improve the achievable competitive ratio. This question is in the spirit of several other methods of analysis that allow an online algorithm to have additional power, such as {\em resource augmentation} \cite{KalyanasundaramP00,PhillipsSTW97}.

The potential of this relaxation is illustrated by the basic secretary problem, where the aim is to select the item of maximum value among randomly ordered inputs. There, it is not difficult to show that if an algorithm  picks every item that is better than the items seen so far, the true maximum will be found, while the expected number of items picked under randomly ordered inputs will be $\log(n)$. Further, we show that this approach can be easily modified to get the maximum with $1-\epsilon$ probability while picking at most  $O(\ln(1/\epsilon))$ items for any constant $\epsilon>0$. Thus, with just a constant size shortlist, we can break the $1/e$ barrier for the secretary problem and achieve a competitive ratio that is arbitrarily close to $1$! 

Motivated by this observation, we ask if a similar improvement can be achieved by relaxing the \nameOfProblem~to allow a shortlist. That is, instead of choosing $k$ items, the algorithm is allowed to chose  $\eta(k)$ items as part of a shortlist, for some function $\eta$;
and at the end of all inputs, the algorithm chooses $k$ items from the $\eta(k)$ selected items. Then, what is the relationship between $\eta(\cdot)$ and the competitive ratio for this problem? Can we achieve a solution close to the best offline solution when $\eta(k)$ is not much bigger than $k$, for example  when 
 $\eta(k)=O(k)$?  

 


In this paper, we answer this question affirmatively by giving a polynomial time algorithm that achieves $1-1/e-\epsilon-O(k^{-1})$ competitive ratio for the \nameOfProblem~using a shortlist of size $\eta(k)=O(k)$. This is surprising since $1-1/e$ is the best achievable approximation (in polynomial time) for the offline problem. Further, for some special cases of submodular functions, we demonstrate that an $O(1)$ shortlist allows us to achieve a $1-\epsilon$ competitive ratio. These results demonstrate the power of (small) shortlists for closing the gap between online and offline (polynomial time) algorithms. 

We also discuss connections of {\nameSecSL} to the related streaming settings. While a streaming algorithm does not qualify as an online algorithm (even when a shortlist is allowed), we show that our algorithm can in fact be implemented in a streaming setting to use $\eta(k)=O(k)$ memory buffer; and our results significantly  improve  the available results for the \streamingProblem.

\subsection{Problem Definition}


We now give a more formal definition.
Items from a set ${\cal U} = \{a_1, a_2, \ldots, a_n\}$  (pool of items) arrive in a uniformly random order over $n$ sequential rounds. The set ${\cal U}$ is apriori fixed but unknown to the algorithm, and the total number of items $n$ is known to the algorithm. In each round, the algorithm irrevocably decides whether to add the arriving item to a {\it shortlist} $A$ or not. 
The algorithm's value at the end of $n$ rounds is given by 
$$\ALG = \Ex[\max_{S\subseteq A, |S|\le k} f(S)]$$ 
where $f(\cdot)$ is a monotone submodular function. The algorithm has value oracle access to this function.

The optimal offline utility is given by
$$\OPT:=f(S^*), \text{ where } S^*=\arg \max_{S \subseteq [n], |S|\le k} f(S).$$ 
We say that an algorithm for this problem achieves a competitive ratio  $c$ using shortlist of size $\eta(k)$, if at the end of $n$ rounds, $|A|\le \eta(k)$ and $\frac{\ALG}{\OPT}\le c$.

Given the shortlist $A$, since the problem of computing the solution $\arg \max_{S\subseteq A, |S|\le k} f(S)$ can itself be computationally intensive, our algorithm will also track and output a subset $A^* \subseteq A, |A^*| \le k$. We will lower bound  the competitive ratio by bounding $\frac{f(A^*)}{f(S^*)}$.

The above problem definition has connections to some existing problems studied in the literature. The well-studied online \nameOfProblem~described earlier is obtained from the above definition by setting $\eta(k)=k$, i.e., it is same as the case when no extra items can be selected as part of a shortlist. 
Another related problem is {\it  \streamingProblem}~studied in~\cite{norouzi}. In this problem, items from a set $\cal U$ arrive online in random order and the algorithm aims to select a subset $S \subseteq {\cal U}, |S|\le k$ in order to maximize $f(S)$. The streaming algorithm is allowed to maintain a {\it buffer} of size $\eta(k)\ge k$. 
However, this streaming problem is distinct from the \nameOfProblemSL\ in several important ways. On one hand, since an item previously selected in the memory buffer can be discarded and replaced by a new items, a memory buffer of size $\eta(k)$ does not imply a shortlist of size at most $\eta(k)$. On the other hand, in the secretary setting, we are allowed to memorize/store more than $\eta(k)$ items without adding them to the shortlist. Thus an algorithm for \nameOfProblem with shortlist of size $\eta(k)$ may potentially use a buffer of size larger than $\eta(k)$. 
Our algorithms, as described in the paper, do use a large buffer, but we will show
that the algorithm presented in this paper can in fact be implemented to use only $\eta(k)=O(k)$ buffer, thus obtaining matching results for the streaming problem. 


\subsection{Our Results}
Our main result is an online algorithm  for \nameOfProblemSL\ that, for any constant $\epsilon>0$, achieves a competitive ratio of $1-\frac{1}{e} - \epsilon-O(\frac{1}{k})$ with $\eta(k) = O(k)$. 
Note that for \nameOfProblem\ there is an upper bound of $1-1/e$ on the achievable aproximation factor, even in the offline setting, and this upper bound applies to our problem for arbitrary size $\eta(\cdot)$ of shortlists. On the other hand for online monotone \nameOfProblem, i.e., when $\eta(k)=k$, the best competitive ratio achieved in the literature is $1/e-O(k^{-1/2})$~\cite{kesselheim} 
Remarkably, with only an $O(k)$ size shortlist, our online algorithm is able to achieve a competitive ratio that is arbitrarily close to the offline upper bound of $1-1/e$.

In the theorem statements below, big-Oh notation $O(\cdot)$ is used to represent asymptotic behavior with respect to $k$ and $n$. We assume the standard  value oracle model:  the only access to the submodular function is through a black box
returning $f(S)$ for a given set $S$, and  each such queary can be done in $O(1)$ time. 
\begin{theorem} \label{opttheorem}
For any constant $\epsilon>0$, there exists an online algorithm (Algorithm \ref{alg:main}) for the \nameOfProblemSL\ that achieves a competitive ratio of $1-\frac{1}{e} -\epsilon -O(\frac{1}{k})$, with shortlist of size $\eta_\epsilon(k)=O(k)$. Here,  $\eta_\epsilon(k)=O(2^{poly(1/\epsilon)}k)$.  
  The running time of this online algorithm is $O(n)$.
\end{theorem}
Specifically, we have $\eta_\epsilon(k)=$\etaMacro$k$ for some constant $c$.
The running time  of our algorithm is linear in $n$, the size of the input, which 
is significant as, until recently, it was not known if there exists a linear time algorithm achieving a $1-1/e-\epsilon$ approximation even for the offline monotone submodular maximization problem under cardinality constraint\cite{linear}. Another interesting aspect of our algorithm is that it is highly parallel. Even though the decision for each arriving item may take time that is exponential in $1/\epsilon$ (roughly  $\eta_\epsilon(k)/k$), it can be readily parallelized among multiple (as many as $\eta_\epsilon(k)/k$) processors. 



Further, we show an implementation of Algorithm 2 that uses a memory buffer of size at most $\eta_\epsilon(k)$ to get the following result for the problem of {\it \streamingProblem} described in the previous section. 
\begin{restatable}{theorem}{thmStreaming}
\label{thm:streaming}
For any constant $\epsilon\in (0,1)$, there exists an algorithm for the \streamingProblem  that achieves $1-\frac{1}{e} -\epsilon -O(\frac{1}{k})$ approximation to $\OPT$ while using a memory buffer of size at most $\eta_\epsilon(k)=O(k)$. Also, the number of objective  function evaluations for each item, amortized over $n$ items, is $O(1+\frac{k^2}{n})$.
\end{restatable}
The above result significantly improves over the state-of-the-art results in random order streaming model \cite{norouzi}, which are an approximation ratio of $\frac 1 2+8\times 10^{-14}$ using a memory of size $O(k\log k)$. 



It is natural to ask whether these $k$-lists are, in fact, too powerful.  Maybe they could actually allow us to always match the best offline algorithm.  We give a negative result in this direction and 
show that even if we have unlimited computation power, for any function $\eta(k)=o(n)$, we can get no better than $7/8$-competitive algorithm using a shortlist of size $\eta(k)$. Note that with unlimited computational power, the offline problem can be solved exactly. This result demonstrates that 
having a shortlist does not make the online problem too easy - 
even with a shortlist (of size $o(n)$) there is an information theoretic gap between the online and offline problem. 

\begin{restatable}{theorem}{hardness}
\label{hardness}
No online algorithm (even with unlimited computational power)   can achieve a competitive ratio better than $7/8+o(1)$ for the \nameOfProblemSL, while using a shortlist of size $\eta(k)=o(n)$.
\end{restatable}
Finally, for some special cases of monotone submodular functions, we can asymptotically approach the optimal solution.
The first one is the family of functions we call $m$-submdular. 
A function $f$ is $m$-submodular if it is submodular and there exists a submodular function $F$ such that for all $S$:
\[ 
f(S)= \max_{T\subseteq S, |T|\le m} F(T) \ .
\]
\scomment{This needs to be rewritten, I don't have time right now, so removing: Example of $m$-submodular functions are   maximum node weighted bipartite matching and  maximum edge weighted bipartite matching defined on $G=(X\times Y)$ with $|Y|=m$.  (the assignments will be done at the end of algorithm after all the selections are made).}

\begin{theorem}
\label{thm:msub}
If $f$ is an $m$-submodular function, there exists an online algorithm for the \nameOfProblemSL~that achieves a competitive ratio of $1-\epsilon$ with shortlist of size $\eta_{\epsilon,m}(k)=O(1)$. Here, 
$\eta_{\epsilon,m}(k) = (2m+3) \ln(2/\epsilon)$.
\end{theorem}
A proof of Theorem~\ref{thm:msub} along with the relevant algorithm (Algorithm \ref{alg:SIII}) appears in the appendix.

Another special case  is  monotone submodular functions $f$ satisfying the following property:
$f(\{a_1,\cdots, a_i+\alpha,\cdots, a_k\})  \ge  f(\{a_1, \cdots, a_i, \cdots, a_k \})$, for any $\alpha >0$ and $1\le i \le k$.
We can show that the algorithm by \citet{kleinberg} asymptotically approaches optimal solution for such functions, but we omit the details.

\subsection{Comparison to related work}
We compare our results (Theorem \ref{opttheorem} and Theorem \ref{thm:streaming}) to the best known results for {\it \nameOfProblem}~and {\it \streamingProblem}, respectively.

The best known algorithm so far for \nameOfProblem~is by \citet{kesselheim}, with asymptotic competitive ratio of $1/e-O(k^{-1/2})$. 
In their algorithm, after observing each element, they use an oracle to compute optimal offline solution on the elements seen so far.
Therefore it requires exponential time in $n$. The best competitive ratio that they can get in polynomial time is 
$\frac{1}{e}(1-\frac{1}{e})-O(k^{-1/2})$.
In comparison, by using a shortlist of size $O(k)$ our (polynomial time) algorithm achieves a competitive ratio of $1-\frac{1}{e}-\epsilon-O(k^{-1})$. In addition to substantially improves the above-mentioned results for \nameOfProblem, this closely matches the best possible offline approximation ratio of $1-1/e$ in polynomial time. Further, our algorithm is linear time. Table \ref{table:t1} summarizes this comparison. 
Here, $O_\epsilon(\cdot)$ hides the dependence on the constant $\epsilon$. The hidden constant in $O_{\epsilon}(.)$ is 
\etaMacro for some absolute constant $c$.

\begin{table}[h!]
\centering
\begin{tabular}{ |c c c c c| }
\hline
 & \#selections & Comp ratio & Running time & Comp ratio in poly(n) \\
\hline
\cite{kesselheim} & $k$ & $1/e-O(k^{-1/2})$ & $exp(n)$ & $\frac{1}{e}(1-1/e)$ \\ 
this & $O_{\epsilon}(k)$ & $1-1/e-\epsilon-O(1/k) $ & $O_{\epsilon}(n)$ &  $1-1/e-\epsilon-O(1/k)$ \\
\hline    

\end{tabular}
\caption{\nameOfProblem~settings}
\label{table:t1}
\end{table}
In the streaming setting, \citet{Chakrabarti2015} provided a single pass streaming algorithm for monotone submodular function maximization under $k$-cardinality constraint, that achieves a $0.25$ approximation under adversarial ordering of input. Further, their algorithm requires $O(k)$ function evaluations per arriving item and $O(k)$ memory.
The currently best known approximation  under  adversarial order streaming model is by~\citet{Badanidiyuru2014StreamingFly}, who achieve a $1/2-\epsilon$ approximation with a memory of size $O(\frac{1}{\epsilon}k\log k)$. 
There is an upper bound of $1/2+o(1)$ on the competitive ratio achievable by any streaming algorithm for this problem under adversarial order, while using $o(n)$ memory~\cite{norouzi}.

\citet{Hess} initiated the study of  \streamingProblem.
Their algorithm uses $O(k)$ memory and a total of $n$ function evaluations to achieve $0.19$ approximation. 
The state of the art result in the random order input model is due to \citet{norouzi} who achieve a $1/2+8\times 10^{-14}$ approximation, while using a memory buffer of size $O(k\log k)$.
Table~\ref{table:t2} provides a detailed comparison of our result in Theorem \ref{thm:streaming} to the 
above-mentioned results for  \streamingProblem, showing that our algorithm substantially improves the existing results on most aspects of the problem. 

\begin{table}[h!]
\centering
\begin{tabular}{ |c c c c c | }
\hline
 & Memory size & Approximation ratio & Running time & update time \\ 
 \hline
\cite{Hess} & $O(k)$ &  $0.19$ & $O(n)$ & O(1) \\ 
\cite{norouzi} & $O(k\log k)$ &  $1/2+8\times 10^{-14}$ & $O(n\log k)$ & $O(\log k)$ \\ 
\cite{Badanidiyuru2014StreamingFly} & $O(\frac{1}{\epsilon}k\log k)$ & $1/2-\epsilon$ & $poly(n,k, 1/\epsilon)$ & $O(\frac{1}{\epsilon}\log k)$ \\ 
this & $O_{\epsilon}(k)$ & $1-1/e-\epsilon-O(1/k)$ & $O_{\epsilon}(n)$ & amortized $O_{\epsilon}(1+\frac{k^2}{n})$ \\
\hline    
\end{tabular}
\caption{\streamingProblem
}
\label{table:t2}
\end{table}

There is also a line of work studying the online variant of the submodular welfare maximization problem (e.g., \cite{vahab,swm,Kapralov:2013}). In this problem, 
the items arrive online, and each arriving item should be allocated  to one of $m$ agents with a submodular valuation functions $w_i(S_i)$ where $S_i$ is the subset of items allocated to $i$-th agent). The goal is to partition the arriving items into $m$ sets to be allocated to $m$ agents, so that the sum of valuations over all agents is maximized. This setting is incomparable with the \nameOfProblem~setting considered here.
 
\subsection{Organization}
The rest of the paper is organized as follows. Section~\ref{sec:alg} describes our main algorithm (Algorithm \ref{alg:main}) for the \nameOfProblemSL, and demonstrates that its shortlist size is bounded by $\eta_\epsilon(k)=O(k)$. In Section \ref{sec:analysis}, we analyze the competitive ratio of this algorithm to prove Theorem~\ref{opttheorem}. In Section \ref{sec:streaming}, we provide an alternate implementation of Algorithm \ref{alg:main} that uses a memory buffer of size at most $\eta_\epsilon(k)$, in order to prove Theorem \ref{thm:streaming}. Finally, in Section \ref{sec:hardness}, we provide a proof of our impossibility result stated in Theorem \ref{hardness}. The proof of Theorem~\ref{thm:msub} along with the relevant algorithm appears in the appendix.

\section{Algorithm description}
\label{sec:alg}

Before giving our algorithm for \nameOfProblemSL, we describe a simple technique for \nameSecSL~
that achieves a  $1-\delta$ competitive ratio for  with shortlists of size logarithmic in $1/\delta$.
Recall that in the \nameSec, the aim is to select an item with expected value close to the maximum among a pool of items $I=(a_1, \ldots, a_N)$ arriving sequentially in a uniformly random order. 
We will consider the variant with shortlists, where we now want to pick a shortlist which contains an item with expected value close to the maximum.
We propose the following simple algorithm. For the first $n\delta/2$ rounds, don't add any items to the shortlist, but just keep track of the maximum value seen so far.
For all subsequent rounds, 
for any arriving item $i$ that has a value $a_i$ greater than or equal to the maximum value seen so far, add it to the shortlist if the size of shortlist is less than or equal to $L=4\ln(2/\delta)$. This algorithm is summarized as Algorithm \ref{alg:SIIImax}. Clearly, for contant $\delta$, this algorithm uses a shortlist of size $L=O(1)$. Further, under a uniform random ordering of input, we can show that the maximum value item will be part of the shortlist with probability $1-\delta$. (See Proposition \ref{maxanalysis} in Section \ref{sec:analysis}.)
\begin{algorithm*}[ht]
  \caption{~\bf{Algorithm for secretary with shortlist} (finding max online)}
  \label{alg:SIIImax} 
\begin{algorithmic}[1]
\State Inputs: number of items $N$, items in $I=\{a_1, \ldots, a_N\}$ arriving sequentially, $\delta \in (0,1]$. 
\State Initialize: $A\leftarrow \emptyset$,  $u=n\delta/2$, $M = -\infty$ 
\State $L \leftarrow 4\ln(2/\delta)$ 

\For {$i= 1$ to $N$}
\If {$ a_i > M$}
\State $M \leftarrow  a_i$
\If {$i \geq u$ and $|A|<L$}
\State $A\leftarrow A\cup \{a_i\}$
\EndIf
\EndIf
\EndFor


\State return $A$, and $A^*:= \max_{i\in A} a_i$
\end{algorithmic}
\end{algorithm*}

\begin{algorithm*}[h!]
  \caption{~\bf{Algorithm for {\bf submodular} $k$-secretary with shortlist}}
  \label{alg:main} 
  \label{alg:tmp} 
\begin{algorithmic}[1]
\State Inputs: set $\bar I=\{\bar{a}_1,\ldots, \bar{a}_n\}$ of $n$ items arriving sequentially, submodular function $f$, parameter $\epsilon \in (0,1]$. 
\State Initialize: $S_0 \leftarrow \emptyset, R_0 \leftarrow \emptyset, A \leftarrow \emptyset, A^* \leftarrow \emptyset$, constants $\alpha \ge 1, \beta \ge 1$ which depend on the constant $\epsilon$.
\State Divide indices $\{1,\ldots, n\}$ into $(\alpha, \beta)$ windows as prescribed by Definition \ref{def:windows}.
\For {window $w= 1, \ldots, k/\alpha$} 

 \For {every slot $s_j$ in window $w$, $j=1,\ldots, \alpha\beta$}
  \State Concurrently for all subsequences of previous slots $\tau\subseteq \{s_1, \ldots, s_{j-1}\}$ of length $|\tau|<\alpha$ \label{li:subb}   \hspace{0.44in} in window $w$, call the online algorithm in Algorithm \ref{alg:SIIImax} with the following inputs: 
  \begin{itemize}[leftmargin=0.7in]
  \item   number of items $N=|s_j|+1$, $\delta=\frac{\epsilon}{2}$, and
  \item item values $I=(a_0, a_1, \ldots, a_{N-1})$, with 
     \begin{eqnarray*} 
  a_0 & := & \max_{x\in R_{1,\ldots, w-1}} \Delta(x|S_{1,\ldots,w-1} \cup \gamma(\tau)) \\
     a_\ell & := & \Delta(s_j(\ell)| S_{1,\ldots,w-1} \cup \gamma(\tau) ),  \forall \ell=1,\ldots, N-1
     \end{eqnarray*}
 where $s_j(\ell)$ denotes the $\ell^{th}$ item in the slot $s_j$. 
  \end{itemize}
\State Let $A_{j}(\tau)$ be the shortlist returned by  Algorithm \ref{alg:SIIImax} for slot $j$ and subsequence $\tau$. Add 
\label{li:tmp} \hspace{0.44in} all items except the dummy item $0$ to the shortlist $A$. 
 That is, \label{li:sube}
 $$A\leftarrow A\cup  (A(j)\cap s_j)$$
 \EndFor
 \State \label{li:Rw} After seeing all items in window $w$, compute $R_w, S_w$ as defined in \eqref{eq:Rw} and \eqref{eq:Sw} respectively.
 \State $A^* \leftarrow A^*\cup (S_w \cap A)$
\EndFor
\State return $A$, $A^*$. 
\end{algorithmic}
\end{algorithm*}


  \toRemove{
  \begin{algorithmic}[1]
  \For {every slot $s_j$ in window $w$, $j=1,\ldots, \alpha\beta$}
  \State Concurrently for all subsequences of previous slots $\tau\subseteq \{s_1, \ldots, s_{j-1}\}$ in window $w$, of length $|\tau|<\alpha$, call the online algorithm in Algorithm \ref{alg:SIIImax} with the following inputs: number of items $N=|s_j|+1$, $\delta = \frac{\epsilon}{{\alpha \beta \choose \alpha}}$, \mcomment{why do you divide by... the errors do not add up} and values of arriving items $(a_0, a_1, \ldots, a_{N-1})$ defined as 
  $$a_0:=\max_{x\in R_{1,\ldots, w-1}} f(S_{1,\ldots,w-1} \cup \gamma(\tau) \cup \{x\}) - f(S_{1,\ldots,w-1} \cup \gamma(\tau)\}$$
    $$a_i :=f(S_{1,\ldots,w-1} \cup \gamma(\tau) \cup \{i\}) - f(S_{1,\ldots,w-1} \cup \gamma(\tau)\})$$
 where $i$ denotes the $i^{th}$ item in slot $s_j$.
  \State Let $A_{j}(\tau)$ be the shortlist returned by  Algorithm \ref{alg:SIIImax} for slot $j$ and subsequence $\tau$. Add all items except the dummy item $0$ to $H_w$, i.e, for all $\tau$,
  $$H_w \leftarrow H_w\cup (A_{j}(\tau) \cap s_j)$$
 \EndFor
\State return $H_w$
\end{algorithmic}
}

There are two main difficulties in extending this idea to the \nameOfProblemSL. First, instead of one item, here we aim to select a set $S$ of $k$ items using an $O(k)$ length shortlist. Second, the contribution of each new item $i$ to the objective value, as given by the submodular function $f$, depends on the set of items selected so far. 

The first main concept we introduce to handle these difficulties  is that of dividing the input into sequential blocks that we refer to as $(\alpha, \beta)$ windows. Below is the precise construction of $(\alpha, \beta)$ windows, for any postivie integers $\alpha$ and $\beta$, such that $k/\alpha$ is an integer.

We use a set of random variables $X_1,\ldots,X_m$ defined in the following way.  Throw $n$ balls into $m$ bins uniformly at random.  Then set $X_j$ to be the number of balls in the $j$th bin.  We call the resulting $X_j$'s a {\em $(n,m)$-ball-bin random set}.

\begin{definition}[$(\alpha,\beta)$ windows]
\label{def:windows}
Let $X_1,\ldots,X_{k\beta}$ be a $(n,k\beta)$-ball-bin random set.
Divide the indices $\{1,\ldots, n\}$ into $k\beta$ slots, where the $j$-th slot, $s_j$, consists of $X_j$ consecutive indices in the natural way, that is, slot $1$ contains the first $X_1$ indices, slot $2$ contains the next $X_2$, etc.
 Next, we define $k/\alpha$ windows, where window $i$ consists of $\alpha \beta$ consecutive slots, in the same manner as we assigned slots.
\end{definition}
Thus, $q^{th}$ slot is composed of indices $\{\ell, \ldots, r\}$, where $\ell=X_1+...+X_{q-1}+1$ and  $r=X_1+...+X_q$. 
Further, if the ordered the input is $\bar a_1, \dots, \bar a_n$ then we say that the items inside the slot $s_q$ are $\bar a_{\ell}, \bar a_{\ell+1}, \dots, \bar a_{r}$ 
To reduce notation, when clear from context, we will use $s_q$ and $w$ to also indicate the {\it set of items} in the slot $s_q$ and window $w$ respectively.

When $\alpha$ and $\beta$ are large enough constants, some useful properties can be obtained from the construction of these windows and slots. First, roughly $\alpha$ items from the optimal set $S^*$ are likely to lie in each of these windows; and further, it is unlikely that two items from $S^*$ will appear in the same slot. (These statements will be made more precise in the analysis where precise setting of $\alpha,\beta$ in terms of $\epsilon$ will be provided.) Consequently, our algorithm can focus on identifying a constant number (roughly $\alpha$) of optimal items from each of these windows, with at most one item coming from each of the $\alpha \beta$ slots in a window. The core of our algorithm is a subroutine that accomplishes this task in an online manner using a shortlist of constant size in each window. 

To implement this idea, we use a greedy selection method that considers all possible $\alpha$ sized subsequences of the $\alpha\beta$ slots in a window, and aims to identify the subsequence that maximizes the increment over the best items identified so far. 
More precisely, for any subsequence $\tau=(s_1,\ldots, s_\ell)$ of the $\alpha\beta$ slots in window $w$, we define a `greedy' subsequence $\gamma(\tau)$ of items as:
\begin{equation}
\label{eq:gamma}
\gamma(\tau):=\{i_1, \ldots, i_\ell\}
\end{equation}
where
\begin{equation}
\label{eq:ij}
i_j := \arg \max_{i\in s_j \cup R_{1, \ldots, w-1}} f(S_{1, \ldots,w-1} \cup \{i_1,\ldots, i_{j-1}\} \cup \{i\}) - f(S_{1, \ldots, w-1} \cup \{i_1\ldots, i_{j-1}\}).
\end{equation}
In ({\ref{eq:ij}) and in the rest of the paper, we use shorthand $S_{1,\ldots, w}$ to denote $S_1\cup \cdots \cup S_{w}$, and $R_{1,\ldots, w}$ to denote $R_1\cup \cdots \cup R_{w}$, etc.  We also will take unions of subsequences, which we interpret as the union of the elements in the subsequences.
We also define $R_w$ to be  the union of all greedy subsequences of length $\alpha$, and $S_w$to be  the best subsequence among those. That is,
\begin{equation}
\label{eq:Rw}
R_w = \cup_{\tau: |\tau|=\alpha} \gamma(\tau)
\end{equation}
and 
\begin{equation}
\label{eq:Sw}
S_w=\gamma(\tau^*),
\end{equation} where
\begin{equation}
\label{eq:taustar}
\tau^*:=\arg \max_{\tau: |\tau|=\alpha} f(S_{1,\ldots,w-1} \cup \gamma(\tau)) - f(S_{1,\ldots,w-1}).
\end{equation}
Note that $i_j$ (refer to \eqref{eq:ij}) can be set as  either an item in slot $s_j$ or an item {\it from a previous greedy subsequence} in $R_1\cup \cdots \cup R_{w-1}$. The significance of the latter relaxation will become clear in the analysis.

As such, identifying the sets $R_w$ and $S_w$ involves looking forward in a slot $s_j$ to find the best item (according to the given criterion in \eqref{eq:ij}) among all the items in the slot. To obtain an online implementation of this procedure, we use an online subroutine that employs the algorithm (Algorithm \ref{alg:SIIImax}) for the basic secretary problem described earlier. This online procedure will result in selection of a set $H_w$ potentially larger than $R_w$, while ensuring that each element from $R_w$ is part of $H_w$ with a high probability $1-\delta$ at the cost of adding extra $\log(1/\delta)$ items to the shortlist. Note that $R_w$ and $S_w$ can be computed {\it exactly at the end} of window $w$. 

Algorithm \ref{alg:tmp} summarizes the overall structure of our algorithm. 
In the algorithm, for any item $i$ and set $V $, we define $\Delta_f(i|V):=f(V\cup \{i\})- f(V)$.

The algorithm returns both the shortlist $A$ which we show to be  of size $O(k)$ in the following proposition, as well as a set $A^*=\cup_w (S_w \cap A)$ of size at most $k$ to compete with $S^*$. In the next section, we will show that $\Ex[f(A^*)] \ge (1-\frac{1}{e} - \epsilon -O(\frac{1}{k})) f(S^*)$ to provide a bound on the competitive ratio of this algorithm. 

\begin{prop}
\label{prop:size}
Given $k,n$, and any constant $\alpha, \beta$ and $\epsilon$, the size of shortlist $A$ selected by Algorithm~\ref{alg:main} is at most $4k \beta {\alpha \beta \choose \alpha}\log(2/\epsilon) = O(k)$. 
\end{prop}
\begin{proof}
For each window $w=1,\ldots, k/\alpha$, and for each of the $\alpha\beta$ slots in this window, lines \ref{li:subb} through \ref{li:sube} in  Algorithm \ref{alg:main} runs  Algorithm \ref{alg:SIIImax} for ${\alpha \beta \choose \alpha}$ times (for all $\alpha$ length subsequences). By construction of  Algorithm \ref{alg:SIIImax}, for each run it will add at most $L\le 4\log(2/\epsilon)$ items each time to the shortlist. 
Therefore, over all windows, Algorithm \ref{alg:main} adds at most $ \frac{k}{\alpha} \times \alpha\beta  {\alpha \beta \choose \alpha} L = O(k)$ items to  the shortlist. 
\end{proof}

\toRemove{
\section{Bounding the size of shortlist}
\label{sec:size}
Given the algorithm description, it is not very difficult to show that 
our algorithm uses a shortlist of size at most $O(k)$ if $\alpha, \beta$ are constants. 
The number of items selected by Algorithm \ref{alg:SIIImax} is at most $L$.
By choosing appropriate size $L$ for shortlist $A$,
w.h.p. we can gaurantee  that $A$ contains the maximum element $M$.
\scomment{Change $\epsilon,\delta$ to get max with probability $1-\epsilon$.}


Now, we can deduce the desired bound on the size of shortlist returned by Algorithm \ref{alg:tmp}
\begin{prop}
The size of shortlist $A$ selected by Algorithm~\ref{alg:tmp} is $O(k)$.
\end{prop}
\begin{proof}
Note that for each window $w=1,\ldots, k/\alpha$, and for each of the $\alpha\beta$ slots in this window, the subroutine in Algorithm \ref{alg:tmpWindow} runs  Algorithm \ref{alg:SIIImax} for ${\alpha \beta \choose \alpha}$ times (for all $\alpha$ length subsequences) to add at most $\log(1/\delta)$ items each time to . Therefore, the above algorithm selects $ k \beta {\alpha \beta \choose \alpha} = O(k)$ items as part of the shortlist.
\end{proof}

}
\toRemove{
\subsection{Bounding the size of shortlist for Algorithm \ref{alg:SIIImax}}
\scomment{Can we add a lemma showing that expected number of items selected will be $\log(1/\delta)$ when $u=N\delta$}
\mcomment{same as previous comment}
\scomment{Also we need to discuss that when we say shortlist of size $\eta(k)$, is that expected number of items or high probability? If it is high probability, we need to define what it means to say shortlist of size $\eta(k)$. Does it mean $\eta(k)$ with prescribed probability $1-\epsilon'$?}
\scomment{the results here will move to the previous section, and proofs for high prob bound will go to appendix}

\begin{theorem}
$\mathbb{E}[|S|] = \ln n$.
\end{theorem}
\begin{proof}
Suppose $S_i=\{a_1, \cdots, a_i\}$.
Consider all permutations of $S_i$, an element will be selected at position $i$ if it is equal to $\max_{a\in S_i} F(a)$, thus the probability that it gets selected is $1/i$. Therefore the expected number of selections will be at most $\sum_{i=1}^{n} \frac{1}{i}  = \ln n $
\end{proof}


\scomment{What is $\delta$ in the theorem below? The proof below seems to suggest it is high probability statement, not exact.}
\begin{theorem}\label{maxanalysis}
Algorithm~\ref{alg:SIIImax}, with parameter $u=n\epsilon$,  selects a set $S$ with 
$$|S|<\ln (1/\epsilon)+\ln(1/\delta)+\sqrt{\ln^2{1/\delta}+2m\ln(1/\delta)\ln(1/\epsilon)}$$ 
and $\mathbb{E}[\max_{a\in S} a]=(1-\epsilon-\delta)OPT$, where $OPT$ is the max element in the input.
\end{theorem}
\begin{proof}
We use Freedman's inequality.
 If $\{a_1,\cdots, a_i\}$ has a unique maximum, 
define $Y_i$ to be a random variable indicating whether the algorithm has selected $a_i$  or not, where $Y_i=1-\frac{1}{i}$ if $a_i$ is selected and $Y_i=-\frac{1}{i}$ otherwise. 
 If it has not unique solution define $Y_i=0$. ($a_i$ will not be selected)
 Also define $\mathcal{F}_i=\{Y_{n},Y_{n-1}, \cdots, Y_{n-i+1}\}$.

Let $X_i=\sum_{j=n-i+1}^{n} Y_j$,
then $\{X_i\}$ is a martingle, 
because $E[X_{i+1}|\mathcal{F}_{i}] = X_i+E[Y_{n-i}|\mathcal{F}_i]$.
 If $\{a_1,\cdots, a_i\}$ has a unique maximum element, $E[Y_{n-i}|\mathcal{F}_i]=(1/i)(1-1/i)+(1-1/i)(-1/i)=0$,
 otherwise $E[Y_{n-i}|\mathcal{F}_i]=0$. So in both cases $E[X_{i+1}|\mathcal{F}_{i}] =X_i$.
As in the Freedman's inequality, let $L=\sum_{i=n\epsilon}^{n} Var(Y_i| F_{i-1})$. 
\begin{align*}
L  =  \sum_{i=n\epsilon}^{n} \frac{1}{i}  (1-\frac{1}{i})^2 + (1-\frac{1}{i}) (\frac{1}{i})^2 
< \sum_{i=n\epsilon}^{n} \frac{1}{i} =\ln (1/\epsilon)
\end{align*}
Therefore,
\[
Pr(X_{n-n\epsilon}\ge \alpha \text{ and }  L\le \ln (1/\epsilon) ) \le exp(-\frac{\alpha^2}{\ln (1/\epsilon)+ 2\alpha })   < \delta 
\]
Thus we get $\alpha > \ln(1/\delta)+\sqrt{\ln^2{1/\delta}+2\ln(1/\delta)\ln(1/\epsilon)}$.
Also $|S| = X_{n-n\epsilon} + \ln(1/\epsilon)$. Therefore 
\[
Pr(|S| \ge  \ln (1/\epsilon)+\ln(1/\delta)+\sqrt{\ln^2{1/\delta}+2\ln(1/\delta)\ln(1/\epsilon)}  )  \le \delta
\]
So with probability $(1-\delta)$, $|S| \le \ln (1/\epsilon)+\ln(1/\delta)+\sqrt{\ln^2{1/\delta}+2\ln(1/\delta)\ln(1/\epsilon)} $.  
Also $\mathbb{P}(OPT \in\{a_{n\epsilon},\cdots, a_n \}  ) =(1-\epsilon)$.
Therefore $E[\max_{a\in S} a] \ge (1-\epsilon)OPT-\delta OPT$. 
\end{proof}

\begin{theorem}
Any online algorithm needs to select at least $\frac{1}{2}\log(1/\epsilon)-\frac{1}{2}$ elements, in expectation, to select the maximum element  with probability at least $(1-\epsilon)$ in a random permutation. (we assume $n> 1/\epsilon$)
\end{theorem}
\begin{proof}
Let $I_i=\{a_1,\cdots, a_{n/2^{i-1}}\}$, $T_i=\{a_{n/2^i+1}, \cdots, a_{n/2^{i-1}}\}$, and $R_i=I_1\setminus I_i$, for $i=1, \cdots, \log(1/\epsilon)$. Suppose $M_i$ is the maximum element in $I_i$. Let $S$  be the set of selected elements by algorithm at the end of execution.
Suppose $\epsilon_i= E[M_i\notin S| M_i\in T_i]$,
then $E[|S\cap T_i|] \ge \frac{1}{2} (1-\epsilon_i)$.
Therefore $E[|S|] \ge \sum_{i=1}^{\log(1/\epsilon)} \frac{1}{2} (1-\epsilon_i) $. Also w.p. $\frac{1}{2^{i}}$, $M_1 \in T_i$, 
thus $\sum_{i=1}^{\log(1/\epsilon)} \frac{1}{2^{i}} \epsilon_i  \le \epsilon$. 
(Note that we use the fact $ E[M_i\notin S| M_i\in T_i \text{ and }  M_i=M_1] \le \epsilon_i$, i.e, if algorithm selects one element it will select 
it even if we increase its value and keep the rest untouched) 
Now $E[|S|]$ is minimized under above constraint if $\frac{1}{2^{\log(1/\epsilon)}}\epsilon_{\log(1/\epsilon)} = \epsilon$ and the rest are zero.
Hence $E[|S|] \ge \frac{1}{2}\log(1/\epsilon)-\frac{1}{2}$.
\end{proof}


}

\section{Bounding the competitive ratio (Proof of Theorem \ref{opttheorem})}
\label{sec:analysis}
\newcommand{\settinga}{\mbox{$k \ge \alpha\beta$}, \settingb}
\newcommand{\settingb}{\mbox{$\beta\ge \frac{8}{(\delta')^2}$}, $\alpha\ge 8\beta^2 \log(1/\delta')$}

In this section we show that for any $\epsilon \in (0,1)$, Algorithm \ref{alg:tmp} with an appropriate choice of constants $\alpha, \beta$, achieves the competitive ratio claimed in Theorem \ref{opttheorem} for the \nameOfProblemSL.  

Recall the following notation defined in the previous section.  For any collection of sets $V_1, \ldots, V_\ell$, we use $V_{1,\ldots, \ell}$ to denote $V_1 \cup \cdots \cup V_\ell$. Also, recall that for any item $i$ and set $V $, we denote $\Delta_f(i|V):=f(V\cup \{i\})- f(V)$.

\paragraph{Proof overview.} The proof is divided into two parts. We first  show a lower bound on the ratio $\Ex[f(\cup_w S_w)]/{\OPT}$ in Proposition~\ref{prop:first}, where $S_w$ is the subset of items as defined in \eqref{eq:Sw} for every window $w$. Later in Proposition~\ref{prop:online}, we use the said  bound to derive a lower bound on the ratio $\Ex[f(A^*)]/\OPT$, where $A^*=A\cap (\cup_w S_w)$ is the subset of shortlist returned by Algorithm~\ref{alg:main}. 

Specifically, in Proposition \ref{prop:first}, we provide settings of parameters $\alpha, \beta$ such that of $\Ex[f(\cup_w S_w)] \ge \left(1-\frac{1}{e}-\frac{\epsilon}{2} - O(\frac{1}{k})\right)\OPT$.  A central idea in the proof of this result is to show that for every window $w$, given $R_{1,\ldots, w-1}$,  the items tracked from the previous windows, any of the $k$ items from the optimal set $S^*$  has at least $\frac{\alpha}{k}$ probability to appear either in window $w$, or among the tracked items $R_{1,\ldots, w-1}$. Further, the items from $S^*$ that appear in window $w$, appear  independently, and in a uniformly at random slot in this window. (See Lemma \ref{lem:pijBound}.)  This observation allows us to show that, in each window, there exists a subsequence $\tilde \tau_w$ of close to $\alpha$ slots, such that the greedy sequence of items $\gamma(\tilde \tau_w)$ will be almost ``as good as" a randomly chosen sequence of $\alpha$ items from $S^*$. More precisely, denoting $\gamma(\tilde \tau_s)=(i_1,\ldots, i_t)$, in Lemma \ref{lem:asGoodas}, for all $j=1,\ldots, t$, we lower bound the increment in function value $f(\cdots)$ on adding $i_j$ over the items in $S_{1, \ldots, w-1} \cup i_{1,\ldots, j-1}$ as:
{\small $$\Ex[\Delta_f(i_j|S_{1,\ldots, w-1} \cup \{i_1, \ldots, i_{j-1}\})|T_{1,\ldots, w-1}, i_1, \ldots, i_{j-1}] \ge \frac{1}{k}\left((1-\frac{\alpha}{k})f(S^*)-f(S_{1,\ldots, w-1} \cup \{i_1, \ldots, i_{j-1}\})\right)\ . $$}
We then deduce (using standard techniques for the analysis of greedy algorithm for submodular functions) that 
{\small $$\Ex[\left(1-\frac{\alpha}{k}\right) f(S^*) - f(S_{1,\ldots, w-1}\cup \gamma(\tilde \tau_w)) | S_{1,\ldots, w-1}] \le e^{-t/k} \left(\left(1-\frac{\alpha}{k}\right) f(S^*)-f(S_{1,\ldots, w-1})\right)\ .$$}
Now, since the length $t$ of $\tilde \tau_w$ is close to $\alpha$ (as we show in Lemma \ref{lem:lengthtau}) and since  $S_w=\gamma(\tau^*)$ with $\tau^*$ defined as the ``best"  subsequence of length $\alpha$ (refer to definition of $\tau^*$ in \eqref{eq:taustar}), we can show that a similar inequality holds for {\small $S_w=\gamma(\tau^*)$, i.e.,
$$\left(1-\frac{\alpha}{k}\right) f(S^*) - \Ex[f(S_{1,\ldots, w-1}\cup S_w)|S_{1,\ldots, w-1}] \le e^{-\alpha/k} \left(1-\delta'\right)\left(\left(1-\frac{\alpha}{k}\right) f(S^*)-f(S_{1,\ldots, w-1})\right)\ ,$$}
where $\delta'\in (0,1)$ depends on the setting of $\alpha, \beta$. (See   Lemma~\ref{lem:Sw}.)  Then repeatedly applying this inequality for $w=1, \ldots, k/\alpha$, and setting $\delta, \alpha, \beta$ appropriately in terms of $\epsilon$, we can obtain $\Ex[f(S_{1,\ldots, W})] \ge \left(1-\frac{1}{e} - \frac{\epsilon}{2} -\frac{1}{k}\right) f(S^*)$, completing  the proof of Proposition \ref{prop:first}

However, a remaining difficulty is that while the algorithm keeps a track of the set $S_w$ for every window $w$, it may not have been able to add all the items in $S_w$ to the shortlist $A$ during the online processing of the inputs in that window. In the proof of Proposition \ref{prop:online}, we show that in fact the algorithm will add most of the items in $\cup_w S_w$ to the short list. More precisely, we show that given that an item $i$ is in $S_w$, it will be in shortlist $A$ with probability $1-\delta$, where $\delta$ is the parameter used while calling Algorithm \ref{alg:SIIImax} in Algorithm \ref{alg:main}. Therefore, using properties of submodular functions it follows  that with $\delta=\epsilon/2$, $\Ex[f(A^*)] = \Ex[f(\cup_w S_w \cap A)] \ge (1-\frac{\epsilon}{2})\Ex[f(\cup_w S_w)]$ (see Proposition \ref{prop:online}). Combining this with the lower bound $\frac{\Ex[f(\cup_w S_w)]}{\OPT} \ge (1-\frac{1}{e}-\frac{\epsilon}{2} - O(\frac{1}{k}))$ mentioned earlier, we complete the proof of competitive ratio bound stated in Theorem \ref{opttheorem}.

\subsection{Preliminaries}
 The following properties of submodular functions are well known (e.g., see~\cite{Buchbinder:2014,Feige:2011,Feldman:2015}). 
\begin{lemma}\label{marginalsum}
Given a monotone submodular function $f$, and subsets $A,B$ in the domain of $f$, we use $\Delta_f(A|B)$ to denote $f(A\cup B)-f(A)$. 
For any set $A$ and $B$, $\Delta_f(A|B) \le \sum_{a\in A\setminus B} \Delta_f(a|B)$
\end{lemma}
\begin{lemma}\label{sample}
Denote by $A(p)$ a random subset of $A$ where each element has probability atleast $p$ to appear in $A$ (not necessarily independently). Then $E[f(A(p))] \ge (1-p) f(\emptyset) + (p)f(A)$
\end{lemma}


We will use the following well known deviation inequality for martingales (or supermartingales/submartingales).

\begin{lemma}[Azuma-Hoeffding inequality]
\label{lem:azuma}
Suppose $\{ X_k : k = 0, 1, 2, 3, ... \}$ is a martingale (or super-martingale) and ${\displaystyle |X_{k}-X_{k-1}|<c_{k},\,} $
almost surely. Then for all positive integers N and all positive reals $r$,
$${\displaystyle P(X_{N}-X_{0}\geq r)\leq \exp \left(\frac{-r^{2}}{2\sum _{k=1}^{N}c_{k}^{2}}\right).} $$
And symmetrically (when $X_k$ is a sub-martingale):
$$
{\displaystyle P(X_{N}-X_{0}\leq -r)\leq \exp \left(\frac{-r^{2}}{2\sum _{k=1}^{N}c_{k}^{2}}\right).} 
$$
\end{lemma}
\begin{lemma}[Chernoff bound for Bernoulli r.v.]
\label{lem:Chernoff}
Let $X = \sum_{i=1}^N X_i$, where $X_i = 1$ with probability $p_i$ and $X_i = 0$ with
probability $1 - p_i$, and all $X_i$ are independent. Let $\mu = \Ex(X) = \sum_{i=1}^N p_i$. Then, 
$$P(X \geq (1 + \delta)\mu) \le e^{{-\delta^2\mu}/{(2+\delta)}} $$
for all $\delta>0$, and
$$P(X \leq (1 - \delta)\mu) \le e^{-\delta^2\mu/2}$$
for all $\delta\in(0,1)$. 
\end{lemma}
\subsection{Some useful properties of $(\alpha, \beta)$ windows}
We first prove some useful properties of $(\alpha, \beta)$ windows, defined in Definition~\ref{def:windows} and used in Algorithm \ref{alg:main}. 
The first observation is that every item will appear uniformly at random in one of the $k\beta$ slots in $(\alpha,\beta)$ windows. 

\begin{definition}
For each item $e\in I$, define $Y_e\in [k\beta]$ as the random variable indicating the slot in which $e$ appears. We call vector $Y\in [k\beta]^n$ a \textit{configuration}.
\end{definition}

\begin{lemma} \label{lem:indep}
Random variables $\{Y_e\}_{e\in I}$ are i.i.d. with uniform distribution on all $k\beta$ slots. 
\end{lemma}

This follows from the uniform random order of arrivals, and the use of the balls in bins process to determine the number of items in a slot during the construction of $(\alpha,\beta)$ windows. A proof is provided in Appendix \ref{app:windows}.

Next, we make important observations about the probability of assignment of items in $S^*$ in the slots in a window $w$, given the sets $R_{1,\ldots, w-1}, S_{1,\ldots, w-1}$ (refer to \eqref{eq:Rw}, \eqref{eq:Sw} for definition of these sets). 
To aid analysis, we define the following new random variable $T_w$ that will track all the useful information from a window $w$.  
\begin{definition}
Define $T_w:= \{(\tau, \gamma(\tau))\}_\tau$, for  all $\alpha$-length subsequences $\tau=(s_1,\ldots, s_\alpha)$ of the $\alpha\beta$ slots in window $w$. Here, $\gamma(\tau)$ is a sequence of items as defined in \eqref{eq:gamma}. Also define  $Supp(T_{1,\cdots ,w}) :=\{e| e\in\gamma(\tau) \text{ for some } (\tau, \gamma(\tau))\in T_{1,\cdots, w} \} $ (Note that $Supp(T_{1,\cdots, w})=R_{1,\ldots, w}$).
\end{definition}

\begin{lemma}
\label{config}
For any window $w\in [W]$,  $T_{1,\ldots, w}$ and $S_{1, \ldots, w}$ \toRemove{are uniquely defined for each configuration $Y$.}
are independent of the ordering of elements within any slot, and 
are determined by the configuration $Y$.
\end{lemma}
\begin{proof}
Given the assignment of items to each slot, it follows from the definition of $\gamma(\tau)$ and $S_w$ (refer to \eqref{eq:gamma} and \eqref{eq:Sw}) that $T_{1,\ldots, w}$ and $S_{1, \ldots, w}$ are independent of the ordering of items within a slot. Now, since  the assignment of items to slot are determined by the configuration $Y$, we obtain the desired lemma statement.
\end{proof}
Following the above lemma, given a configuration $Y$, we will some times use the notation $T_{1,\ldots, w}(Y)$ and $S_{1, \ldots, w}(Y)$ to make this mapping explicit.

\begin{lemma}
\label{lem:pijBound}
For any item $i\in S^*$, window $w \in \{1,\ldots, W\}$, and slot $s$ in window $w$, define
\begin{equation}
\label{eq:pijBound}
p_{is}:=\mathbb{P}(i \in s \cup Supp(T) | T_{1,\ldots, w-1}=T).
\end{equation}
Then, for any pair of slots $s',s''$ in windows $w, w+1, \ldots, W$,
\begin{equation}
p_{is'}=p_{is''} \ge \frac{1}{k\beta} \ .
\end{equation}
\end{lemma}
\begin{proof}
If $i\in Supp(T)$ then the statement of the lemma is trivial, so consider $i\notin Supp(T)$. For such $i$, $p_{is}=\mathbb{P} (Y_i=s  | T_{1,\ldots, w-1}=T)$. 

We show that for any pair of slots $s,s'$, where $s$ is a slot in first $w-1$ windows and $s'$ is a slot in window $w$, 
\begin{equation}\label{eq1}
\mathbb{P}(T_{1,\ldots, w-1}=T|Y_i=s) \le \mathbb{P}(T_{1,\ldots, w-1}=T|Y_i=s') \ .
\end{equation}
And, for any 
pair of slots $s', s''$ in windows $\{w, w+1 ,\cdots, W\}$, 
\begin{equation}
\label{eq:afterw}
\mathbb{P} ( T_{1,\ldots, w-1}=T | Y_i=s' ) =\mathbb{P} ( T_{1,\ldots, w-1}=T | Y_i=s'').
\end{equation}
To see \eqref{eq1}, suppose for a configuration $Y$ we have $Y_i=s$ and $T_{1,\cdots, w-1}(Y)=T $. 
Since $i\notin Supp(T)$, then by definition of $T_{1,\ldots, w-1}$ we have that $i\notin \gamma(\tau)$ for any  $\alpha$ length subsequence $\tau$ of slots in any of the windows $1,\ldots, w-1$. Therefore, if we remove $i$ from windows ${1,\cdots, w-1}$ (i.e., consider another configuration where $Y_i$ is in windows $\{w, \ldots, W\}$) then $T_{1,\cdots, w-1}$ would not change. This is because $i$ is not the output of argmax in definition of $\gamma(\tau)$ (refer to  \eqref{eq:gamma}) for any $\tau$, so that its removal will not change the output of argmax. 
Also by adding $i$ to slot $s'$, $T_{1,\cdots, w-1}$ will not change since $s'$ is not in window $1,\cdots, w-1$.
Suppose configuration $Y'$ is a new configuration obtained from $Y$ by changing $Y_i$ from $s$ to $s'$. 
Therefore $T_{1,\cdots ,w-1}(Y') = T$. 
Also remember that from lemma~\ref{eqprob}, $\mathbb{P} (Y) = \mathbb{P}(Y')$.
This mapping shows that $\mathbb{P}(T_{1,\ldots, w-1}=T|Y_i=s) \le \mathbb{P}(T_{1,\ldots, w-1}=T|Y_i=s')$. 

The proof for \eqref{eq:afterw} is similar.

By applying Bayes' rule to \eqref{eq1} we have 
\[
\mathbb{P} (Y_i=s  | T_{1,\ldots, w-1}=T) \frac{ \mathbb{P}(T_{1,\ldots, w-1}=T) }{\mathbb{P}(Y_i=s) } 
\le \mathbb{P} (Y_i=s'  | T_{1,\ldots, w-1}=T) \frac{ \mathbb{P}(T_{1,\ldots, w-1}=T) }{\mathbb{P}(Y_i=s') } \ .
\]
Also from Lemma~\ref{lem:indep}, $\mathbb{P}(Y_i=s)  = \mathbb{P}(Y_i=s')$ thus 
\[
\mathbb{P} (Y_i=s  | T_{1,\ldots, w-1}=T) 
\le \mathbb{P} (Y_i=s'| T_{1,\ldots, w-1}=T)  \ .
\]
Now, for any pair of slots $s',s''$ in  windows $w, w+1, \cdots, W$, by applying Bayes' rule to the equation \eqref{eq:afterw},
we have $p_{is'}=\mathbb{P} (Y_i=s'  | T_{1,\ldots, w-1}=T) =\mathbb{P} (Y_i=s''  | T_{1,\ldots, w-1}=T)=p_{is''}$.
That is, $i$ has as much probability to appear in $s'$ or $s''$ as any of the other (at most $k\beta$) slots in windows $w, w+1, \ldots, W$. 
As a result $p_{is''}=p_{is'} \ge \frac{1}{k\beta}$.

\end{proof}

\begin{lemma}
\label{lem:ijindep}
For any window $w$,  $i,j\in S^*, i\ne j$ and $s,s'\in w$, the random variables $\mathbf{1}(Y_i=s|T_{1,\cdots, w-1}=T)$ and $ \mathbf{1}(Y_j=s'|T_{1,\cdots, w-1}=T)$ are independent. That is, given $T_{1,\cdots, w-1}=T$,  items $i,j\in S^*, i\ne j$ appear in any slot $s$ in $w$ independently.
\end{lemma}
\begin{proof}
To prove this, we show that $\mathbb{P}(Y_i=s|T_{1,\cdots, w-1}=T)=\mathbb{P}(Y_i=s|T_{1,\cdots, w-1}=T \text{ and } Y_j =s')$. Suppose $Y'$ is a configuration such that $Y'_i=s$ and $Y'_j=s'$, and $T_{1,\cdots, w-1}(Y')=T$. Assume there exists another feasible slot assignment of $j$, i.e., there is another configuration $Y''$ such that $T_{1,\cdots, w-1} (Y'') =T$ and $Y''_j=s''$ where $s''\ne s'$. (If no such configuration $Y''$ exists, then $\mathbf{1}(Y_j=s')|T$ is always $1$, and the desired lemma statement is trivially true.) Then, we prove the desired independence by showing that there exists a feasible configuration where slot assignment of $i$ is $s$, and $j$ is $s''$. This is obtained by changing $Y_j$ from $s'$ to $s''$ in $Y'$, to obtain another configuration $\bar{Y}$. In Lemma~\ref{addtoslot}, we show that this change will not effect $T_{1,\cdots, w-1}$, i.e., $T_{1,\cdots, w-1} (\bar Y)=T $. 
Thus configuration $\bar Y$ satisfies the desired statement.
\end{proof}

\begin{lemma} \label{addtoslot}
Fix a slot $s'$, $T$, and $j\notin Supp(T)$. Suppose that there exists  some configuration  $Y'$ such that $T_{1,\cdots, w-1} (Y') =T$ and $Y_j'=s'$. Then, given any configuration $Y''$ with $T_{1,\ldots, w-1}(Y'')=T$, we can replace  $Y''_j$ with $s'$ to obtain a new configuration $\bar Y$ that also satisfies $T_{1,\ldots, w-1}(\bar Y)=T$.
\toRemove{
Suppose there exists at least one configuration $\bar{Y}$ such that   $T_{1,\cdots, w-1} (\bar{Y}) =T$ and $\bar{Y}_j=s$ for $j\notin Supp(T)$.
Then for any configuration $Y$ with  $T_{1,\cdots, w-1} (Y) =T$, by setting $Y_j=s$, we get a new configuration $Y'$
such that $T_{1,\cdots, w-1} (Y') =T$.}
\end{lemma}
\begin{proof}
Suppose the slot $s'$ lies in window $w'$.
If $w' \ge w$ then the statement is trivial. So suppose $w' < w$.
Create an intermediate configuration by {\it removing} the item $j$ from $Y''$, call it $Y^-$. Since $j\notin Supp(T_{1,\cdots, w-1}(Y'')) = Supp(T)$ we have $T_{1,\cdots, w-1} (Y^-) =T$. In fact, for every subsequence $\tau$, the greedy subsequence for $Y''$, will be same as that for $Y^-$, i.e.,   $\gamma_{Y''}(\tau) = \gamma_{Y^-}(\tau)$.
Now add item $j$ to slot $s'$ in $Y^-$, to obtain configuration $\bar Y$. We claim $T_{1,\cdots, w-1} (\bar Y) =T$.

By construction of $T_{1,\ldots, w}$, we only need to show that $j$ will not be part of the greedy subsequence $\gamma_{\bar Y}(\tau)$ for any subsequence $\tau, |\tau| = \alpha$ containing the slot $s'$ when the input is in configuration $\bar Y$. To prove by contradiction, suppose that $j$ is part of greedy subsequence for some $\tau$ ending in the slot $s'$. \toRemove{We only need to show that $j$ will not get selected in slot $s$ in configuration $Y'$.
Suppose $j$ gets selected in slot $s$ for some subsequence $\tau$ ending in slot $s$.}
\toRemove{Suppose that $\gamma_{Y^-}(\tau)  := \{i_1, \cdots, i_{t-1}, i_t\} = \gamma_{Y''}(\tau) $.} 
For this $\tau$, let $\gamma_{Y^-}(\tau)  := \{i_1, \cdots, i_{\alpha-1}, i_\alpha\} = \gamma_{Y''}(\tau) $. Note that since the items in the  slots  before $s'$ are identical for $\bar Y$ and $Y^-$, \toRemove{and $T_{1,\cdots, w-1} (Y^-) =T$,} we must have that $\gamma_{\bar Y}(\tau) = \{i_1, \cdots, i_{\alpha-1}, j\}$, i.e.,
$\Delta_f(j | S_{1,\ldots,w'-1} \cup \{i_1,\ldots, i_{\alpha-1}\} ) \ge \Delta_f(i_\alpha | S_{1,\ldots,w'-1} \cup \{i_1,\ldots, i_{\alpha-1}\} ) $.
On the other hand, since $T_{1,\cdots, w'-1} (Y') = T_{1, \cdots, w'-1}(Y'') = T (\text{restricted to $w'-1$ windows})$,  we have that $\gamma_{Y'} (\tau) =\{i_1, \cdots, i_\alpha\}$.
However, $Y'_j=s'$. Therefore $j$ was not part of the greedy subsequence $\gamma_{Y'}(\tau)$ even though it was in the last slot in $\tau$, implying $\Delta_f(j | S_{1,\ldots,w'-1} \cup \{i_1,\ldots, i_{t-1}\} ) < \Delta_f(i_t | S_{1,\ldots,w'-1} \cup \{i_1,\ldots, i_{t-1}\} ) $. This  contradicts the earlier observation.

\end{proof}

\subsection{Bounding $\Ex[f(\cup_w S_w)]/\OPT$}
 


In this section, we use the observations from the previous sections to show the existence of a random subsequence of slots $\tilde \tau_w$ of window $w$ such that we can lower bound $f(S_{1,\ldots, w-1}\cup \gamma(\tilde \tau_w))- f(S_{1,\ldots, w-1})$ in terms of $\OPT -f(S_{1,\ldots, w-1})$. This will be used to lower bound  increment $\Delta_f(S_w|S_{1,\ldots, w-1}) = f(S_{1,\ldots, w-1}\cup \gamma(\tau^*)) -f(S_{1,\ldots, w-1})$ in every window.

\begin{definition}[$Z_s$ and $\tilde \gamma_w$]
\label{def:tauw}
Create sets of items $Z_s, \forall s\in w$  as follows: for every slot $s$, add every item from $i\in S^*\cap s$ independently with probability $\frac{1}{k \beta p_{is}}$ to $Z_s$. Then, for every item $i\in S^*\cap T$, with probability $\alpha/k$, add $i$ to $Z_s$ for a randomly chosen slot $s$ in $w$. Define subsequence $\tilde \tau_w$ as the sequence of 
slots with $Z_s\ne \emptyset$. 

\end{definition}
\begin{lemma}
\label{lem:Zs}
Given any $T_{1,\ldots, w-1}=T$, for any slot $s$ in window $w$, all $i, i' \in S^*, i\ne i'$ will appear in $Z_s$ independently with probability $\frac{1}{k\beta}$. Also, given $T$, for every $i \in S^*$, the probability to appear in $Z_s$ is equal for all slots $s$ in window $w$. Further, each $i\in S^*$ occurs in  $Z_s$ of at most one slot $s$.
\end{lemma}
\begin{proof}
First consider $i\in S^*\cap Supp(T)$. Then, $\Pr(i\in Z_s|T) = \frac{\alpha}{k}\times \frac{1}{\alpha \beta} = \frac{1}{k\beta}$ by construction. Also, the event $i\in Z_s|T$ is independent from $i'\in Z_s|T$ for any $i'\in S^*$ as $i$ is independently assigned to a $Z_s$ in construction. Further, every $\in S^*\cap T$ is assigned with equal probability to every slot in $s$.

Now, consider $i\in S^*, i\notin Supp(T)$. Then, for all slots $s$ in window $w$,
$$\Pr(i \in Z_s|T) =\Pr(Y_i=s | T) \frac{1}{p_{is}k\beta} =  p_{is}\times \frac{1}{p_{is}k\beta} = \frac{1}{k\beta},$$
where $p_{is}$ is defined in \eqref{eq:pijBound}. We used that 
$p_{is}=\Pr(Y_i=s|T)$ for $i\notin Supp(T)$. 
Independence of events  $i\in Z_s|T$ for items in $S^*\backslash Supp(T)$ follows from Lemma \ref{lem:ijindep}, which ensures $Y_i=s|T$ and  $Y_j=s|T$ are independent for $i\ne j$; and from independent selection among items with $Y_i=s$ into $Z_s$. 

The fact that every $i\in S^*$ occurs in at most one $Z_s$ follows from construction: $i$ is assigned to $Z_s$ of only one slot if $i\in Supp(T)$; and for $i\notin Supp(T)$, it can only appear in $Z_s$ if $i$ appears in slot $s$.
\end{proof}



\begin{lemma}
\label{lem:asGoodas}
Given the sequence $\tilde \tau_w=(s_1, \ldots, s_t)$ defined in Definition \ref{def:tauw}, let $\gamma(\tilde \tau_s)=(i_1,\ldots, i_t)$, with $\gamma(\cdot)$ as defined in \eqref{eq:gamma}. Then, for all $j=1,\ldots, t$,
{\small
$$\Ex[\Delta_f(i_j|S_{1,\ldots, w-1} \cup \{i_1, \ldots, i_{j-1}\})|T_{1,\ldots, w-1}, i_{1,\ldots, j-1}] \ge \frac{1}{k}\left((1-\frac{\alpha}{k})f(S^*)-f(S_{1,\ldots, w-1} \cup \{i_1, \ldots, i_{j-1}\})\right)\ .$$}
\end{lemma}
\begin{proof}
For any slot $s'$ in window $w$, let $\{s:s \succ_w s'\}$ denote all the slots $s'$ in the sequence of slots in window $w$. 

Now, using Lemma \ref{lem:Zs}, for any slot $s$ such that $s \succ_w s_{j-1}$, 
we have that the random variables  $\mathbf{1}(i\in Z_s | Z_{s_1} \cup \ldots \cup Z_{s_{j-1}})$ are i.i.d. for all $i \in S^*\backslash \{Z_{s_1} \cup \ldots \cup Z_{s_{j-1}}\}$. 
Next, we show that the probabilities $\Pr(i\in Z_{s_j}| Z_{s_1} \cup \ldots \cup Z_{s_{j-1}})$ are identical for all $i\in S^*\backslash \{Z_{s_1} \cup \ldots \cup Z_{s_{j-1}}\}$:
{\small \begin{eqnarray*}
\Pr(i\in Z_{s_j}| Z_{s_1} \cup \ldots \cup Z_{s_{j-1}}) & = & \sum_{s:s\succ_w s_{j-1}} \Pr(i\in Z_s, s=s_j | Z_{s_1} \cup \ldots \cup Z_{s_{j-1}})\\
& = & \sum_{s: s\succ_w s_{j-1}} \Pr(i\in Z_s| s=s_j, Z_{s_1} \cup \ldots \cup Z_{s_{j-1}}) \Pr(s=s_j|Z_{s_1} \cup \ldots \cup Z_{s_{j-1}}) \ .\\
\end{eqnarray*}}
Now, from Lemma \ref{lem:Zs}, the probability $\Pr(i\in Z_s| s=s_j, Z_{s_1} \cup \ldots \cup Z_{s_{j-1}}) $ must be  identical for all $i\notin  Z_{s_1} \cup \ldots \cup Z_{s_{j-1}}$. Therefore, from above we have that for all $i, i' \in S^*\backslash \{Z_{s_1} \cup \ldots \cup Z_{s_{j-1}}\}$,
\begin{equation} 
\label{eq:pk}
\Pr(i\in Z_{s_j}| Z_{s_1} \cup \ldots \cup Z_{s_{j-1}}) = \Pr(i'\in Z_{s_j}| Z_{s_1} \cup \ldots \cup Z_{s_{j-1}})  \ge \frac{1}{k} \ .
\end{equation}
The lower bound of $1/k$ followed from the fact that at least one of the items from $S^*\backslash \{Z_{s_1} \cup \ldots \cup Z_{s_{j-1}}\}$ must appear in $Z_{s_j}$ for $s_j$ to be included in $\tilde \tau_w$. Thus, each of these probabilities is at least $1/k$. 
In other words, if an item is randomly picked from $Z_{s_j}$, it will be $i$ with probability at least $1/k$, for all $i\in S^*\backslash \{Z_{s_1} \cup \ldots \cup Z_{s_{j-1}}\}$.

Now, by definition of $\gamma(\cdot)$ (refer to \eqref{eq:gamma}), $i_j$ is chosen greedily to maximize the increment $\Delta_f(i|S_{1,\ldots, w-1} \cup i_{1,\ldots, s-1})$ over all $i\in s_j \cup Supp(T_{1,\ldots, w-1}) \supseteq Z_{s_j}$. Therefore, we can lower bound the increment provided by $i_j$ by that provided by a randomly picked item from $Z_{s_j}$.
\begin{eqnarray*} 
& & \Ex[\Delta_f(i_j|S_{1,\ldots, w-1} \cup \{i_1, \ldots, i_{j-1}\}|T_{1,\ldots, w-1}=T, i_1, \ldots, i_{j-1}] \\
(\text{using }\eqref{eq:pk}) & \ge & \frac{1}{k} \Ex[\sum_{i\in S^*\backslash \{Z_1,\ldots Z_{s_{j-1}}\}} \Ex[\Delta_f(i|S_{1,\ldots, w-1} \cup \{i_1, \ldots, i_{j-1}\}|T, i_1, \ldots, i_{j-1}]] \\
(\text{using  Lemma~\ref{marginalsum}, monotonicity of $f$} ) & \ge & \frac{1}{k} \Ex[ \left(f(S^*\backslash  \{Z_1,\ldots Z_{s_{j-1}}\})-f(S_{1,\ldots, w-1} \cup i_{1,\ldots, s-1})\right) | T]\\
(\text{using  monotonicity of }f) & \ge & \frac{1}{k}  \Ex[\left(f(S^*\backslash  \cup_{s'\in w} Z_{s'})-f(S_{1,\ldots, w-1} \cup i_{1,\ldots, s-1})\right)| T]\\
(\text{using  Lemma~\ref{lem:Zs} and Lemma~\ref{sample}})  & \geq & \frac{1}{k} \left(\left(1-\frac{\alpha}{k}\right)f(S^*) -f(S_{1,\ldots, w-1} \cup i_{1,\ldots, s-1})\right)
\end{eqnarray*}
The last inequality uses the observation from Lemma \ref{lem:Zs} that given $T$, every $i\in S^*$ appears in $\cup_{s'\in w} Z_{s'}$ independently with probability $\alpha/k$, so that every $i\in S^*$ appears in $S^*\backslash\cup_{s'\in w} Z_{s'}$ independently with probability $1-\frac{\alpha}{k}$; along with Lemma \ref{sample} for submodular function $f$. 
\end{proof}
Using standard techniques for the analysis of greedy algorithm, the following corollary of the previous lemma can be derived: given any $T_{1,\ldots, w-1}=T$:
\begin{lemma}
\label{cor:asGoodas}
$$\Ex\left[\left(1-\frac{\alpha}{k}\right) f(S^*) - f(S_{1,\ldots, w-1}\cup \gamma(\tilde \tau_w)) | T\right]\le \Ex\left[e^{-\frac{|\tilde \tau_w|}{k}} \left|\right. T\right] \left(\left(1-\frac{\alpha}{k}\right) f(S^*)-f(S_{1,\ldots, w-1})\right)$$
\end{lemma}
\begin{proof}
Let $\pi_0= (1-\frac{\alpha}{k}) f(S^*) - \Ex[f(S_{1,\ldots, w-1}) | T_{1,\ldots, w-1}=T]$, and for $j \ge 1$,
$$\pi_j:=(1-\frac{\alpha}{k}) f(S^*) - \Ex[f(S_{1,\ldots, w-1}\cup \{i_1, \ldots, i_j\}) | T_{1,\ldots, w-1}=T, i_1,\ldots, i_{j-1}],$$
Then, subtracting and adding $(1-\frac{\alpha}{k}) f(S^*)$ from the left hand side of the previous lemma, and taking expectation conditional on $T_{1,\ldots, w-1}=T, i_1, \ldots, i_{j-2}$, we get
$$ - \Ex[\pi_{j} | T, i_1, \ldots, i_{j-2}] + \pi_{j-1} \ge \frac{1}{k} \pi_{j-1}$$
which implies
$$\Ex[\pi_j|T, i_1, \ldots, i_{j-2}] \le \left(1-\frac{1}{k}\right) \pi_{j-1} \le \left(1-\frac{1}{k}\right)^j \pi_0\ .$$
By martingale stopping theorem, this implies:
$$\Ex[\pi_t|T] \le \Ex\left[\left(1-\frac{1}{k}\right)^t \left| T\right. \right] \pi_0 \le \Ex\left[e^{-t/k}| T\right] \pi_0\ .$$
where stopping time $t=|\tilde{\tau}_w|$. ($t=|\tilde \tau_w| \le \alpha\beta$ is bounded, therefore, martingale stopping theorem can be applied).

\end{proof}

Next, we compare $\gamma(\tilde \tau_w)$ to $S_w=\gamma(\tau^*)$ . Here, $\tau^*$ was defined has the `best' greedy subsequence of length $\alpha$ (refer to \eqref{eq:Sw} and \eqref{eq:taustar}). To compare it with $\tilde \tau_w$, we need a bound on size of $\tilde \tau_w$. 
\begin{lemma}
\label{lem:lengthtau}
For any real $\delta\in (0,1)$, 
and if $k \ge \alpha\beta$, $\alpha \ge 8\log(\beta)$ and $\beta \ge 8$, 
then given any $T_{1,\ldots, w-1}=T$,
$$(1-\delta)\left(1-\frac{4}{\beta}\right)\alpha \le |\tilde \tau_w| \le (1+\delta)\alpha,$$
with probability $1-\exp(-\frac{\delta^2\alpha}{8\beta})$.
\end{lemma}
\begin{proof}
By definition,
$$|\tilde \tau_w| = |s\in w: Z_s\ne \phi|\ .$$
Again, we use $s'\prec_w s$ to denote all slots before $s$ in window $w$. Then, from Lemma \ref{lem:Zs}, given $T_{1,\ldots, w-1}=T$, for all $i\cap S^*$ and slot $s$ in window $w$, $\Pr[i\in Z_s | Z_{s'}, s'\prec_w s, T]$ is either $0$ or $1/(k\beta)$. Therefore,
$$\Pr[Z_s \ne \phi | T, Z_{s'}, s'\prec_w s]\le \sum_{i\in S^*}  \frac{1}{k\beta} = \frac{1}{\beta}\ .$$
Therefore $X_s=|s'\preceq_w s: Z_{s'}\ne \phi| - \frac{s}{\beta}$ is a super-martingale, with $X_s-X_{s-1}\le 1$. Since there are $\alpha \beta$ slots in window $w$, $X_{\alpha\beta}=|s\in w: Z_{s}\ne \phi| - \alpha$.  Applying Azuma-Hoeffding inequality to  $X_{\alpha\beta}$ (refer to Lemma \ref{lem:azuma}) we get that 
\begin{equation}
\label{eq:upper}
\Pr\left( |s\in w: Z_s \ne \phi| \ge (1+\delta) \alpha |T\right) \le \exp\left(-\frac{\delta^2\alpha}{2\beta}\right)
\end{equation}
which proves the desired upper bound.

For lower bound, first observe that every $i\in S^*$ appears in $\cup_{s\in w} Z_s$ independently with probability $\frac{\alpha}{k}$. Using Chernoff bound for Bernoulli random variables (Lemma \ref{lem:Chernoff}),  for any $\delta\in(0,1)$
\begin{equation}
\label{eq:lower1}
\Pr(||\cup_{s\in w}Z_s| -\alpha| > \delta\alpha) \le \exp(-\delta^2\alpha/3) \ .
\end{equation}

Also, from independence of $i\in Z_s|T$ and  $i'\in Z_s|T$ for any $i,i'\in S^*, i\ne i'$ (refer to Lemma \ref{lem:Zs}),
$$\Pr(i,i'\in Z_s|T, i,i'\notin Z_{s'} \text{ for any } s'\prec_w s) \le  \frac{1}{k^2\beta^2}$$
for any $s\in w$; so that
\begin{equation}
\label{eq:lower10}
\Pr\left(|Z_s|=1|T,  Z_{s'}, s'\prec_w s\right) \ge  \frac{k-|Z_{s'}: s'\prec_w s|}{k\beta}-\frac{1}{\beta^2} \ge \left(1-\frac{2\alpha}{k} \right)\frac{1}{\beta} - \frac{1}{\beta^2}- e^{-\frac{\alpha}{4}} =:p \ .
\end{equation}
where in the last inequality we substituted the upper bound  on $ |Z_{s'}: s'\prec_w s|$ from \eqref{eq:lower1}. 
Specifically, using  \eqref{eq:lower1} with $\delta=3/4$, we obtained  that $|Z_{s'}: s'\prec_w s|\le (1+\frac{3}{4})\alpha \le 2\alpha$ with probability $\exp(-\alpha/4)$. 
Also if $\alpha \ge  8 \log(\beta)$, and $k\ge \alpha\beta$,  we have $ p :=\left(1-\frac{2\alpha}{k} - \frac{1}{\beta}\right)\frac{1}{\beta} - e^{-\frac{\alpha}{4}} \ge  (1-\frac{4}{\beta})\frac{1}{\beta}$.

Now, applying Azuma-Hoeffding inequality (Lemma \ref{lem:azuma}), the total number of slots (out of $\alpha\beta$ slots) for which $|Z_s|=1$ can be lower bounded by:



\begin{equation}
\label{eq:lower2}
\Pr\left(|\{s \in w:|Z_s|=1\}| \ge (1-\delta)p\alpha\beta|T\right) \le \exp\left(-\frac{\delta^2 p^2 \alpha \beta}{2}\right) \ .
\end{equation}
 Substituting $p\ge  (1-\frac{4}{\beta})\frac{1}{\beta}$,
$$\Pr\left(|\{s \in w:|Z_s|=1\}| \ge (1-\delta)(1-\frac{4}{\beta})\alpha|T\right) \le \exp\left(-\frac{\delta^2 (1-4/\beta)^2 \alpha}{2\beta}\right)\ . $$
We further substitute $\beta \ge 8$ in the right hand side of the above inequality, to bound the probability by $\exp(-\delta^2\alpha/8\beta)$.


\end{proof}

\begin{lemma}[Corollary of Lemma \ref{lem:lengthtau}]
\label{cor:lengthtau}
For any real $\delta'\in (0,1)$,  if parameters $k,\alpha, \beta$ satisfy \settinga, then given any $T_{1,\ldots, w-1}=T$, with probability at least $1-\delta' e^{-\alpha/k}$,
$$|\tilde \tau_w| \ge (1-\delta') \alpha\ . $$

\end{lemma}
\begin{proof}
We use the previous lemma with $\delta=\delta'/2$ to get lower bound of $(1-\delta')\alpha$ with probability $1-\exp(-(\delta')^2\alpha/32\beta)$. Then, substituting 
$ k\ge \alpha\beta \ge \frac{64\beta}{(\delta')^2} \log(1/\delta')$ so that using $\beta \le \frac{k(\delta')^2}{64 \log(1/\delta')}$ we can bound the violation probability by
$$\exp(-(\delta')^2 \alpha/32\beta)\le \exp(-(\delta')^2 \alpha/64\beta)\exp(-\alpha/k) \le \delta' e^{-\alpha/k}.$$
where the last inequality uses $\alpha\ge 8\beta^2 \log(1/\delta')$ and $\beta \ge 8/(\delta')^2$.
\end{proof}

\begin{lemma}
\label{lem:Sw}
For any real $\delta'\in (0,1)$, if parameters 
$k,\alpha, \beta$ satisfy \settinga, then
$$\Ex\left[\frac{k-\alpha}{k} \OPT -f(S_{1,\ldots, w})|T_{1,\ldots, w-1}\right] \le (1-\delta') e^{-\alpha/k} \left(\frac{k-\alpha}{k} \OPT - f(S_{1,\ldots, w-1})\right)\ .$$
\end{lemma}
\begin{proof}
The lemma follows from substituting Lemma \ref{cor:lengthtau} in Lemma \ref{cor:asGoodas}.
\end{proof}

Now, we can deduce the following proposition. 
\begin{prop}
\label{prop:first}
For any real $\delta'\in (0,1)$, if parameters 
$k,\alpha, \beta$ satisfy \settinga, then the set $S_{1,\ldots, W}$ tracked by Algorithm \ref{alg:main} satisfies
$$\mathbb{E}[f(S_{1,\ldots, W})] \ge (1-\delta')^2(1-1/e)  \OPT.$$
\end{prop}
\begin{proof}
By multiplying the inequality Lemma \ref{lem:Sw} from $w=1, \ldots, W$, where $W=k/\alpha$, we get 
$$\mathbb{E}[f(S_{1,\ldots, W})] \ge (1-\delta')(1-1/e) (1-\frac{\alpha}{k}) \OPT.$$
Then, using $1-\frac{\alpha}{k}\ge 1-\delta'$ because $k\ge \alpha \beta \ge \frac{\alpha}{\delta'}$, we obtain the desired statement.
\end{proof}



\subsection{Bounding $\Ex[f(A^*)]/\OPT$}
Here, we compare $f(S_{1\ldots, W})$ to $f(A^*)$, where $A^*=S_{1\ldots, W} \cap A$, with $A$ being the shortlist returned by Algorithm \ref{alg:main}. The main difference between the two sets is that in construction of shortlist $A$, Algorithm \ref{alg:SIIImax} is being used to compute the argmax in the definition of $\gamma(\tau)$, in an online manner. This argmax may not be computed exactly, so that some items from $S_{1\ldots, W}$ may not be part of the shortlist $A$. We use the following guarantee for Algorithm~\ref{alg:SIIImax} to bound the probability of this event.



\begin{restatable}{prop}{maxanalysis}
\label{maxanalysis}
For any $\delta\in (0,1)$, and input $I=(a_1,\ldots, a_N)$, Algorithm~\ref{alg:SIIImax}
returns $A^*=\max(a_1,\ldots, a_N)$ with 
probability $(1-\delta)$. 
\end{restatable}

The proof of the above proposition appears in Appendix \ref{app:msubm}. Intuitively, it follows from the observation that if we select every item that improves the maximum of items seen so far, we would have selected $\log(N)$ items in expectation. The exact proof involves showing that on waiting $n\delta/2$ steps and then selecting maximum of every item that improves the maximum of items seen so far, we miss the maximum item with at most $\delta$ probability, and select at most $O(\log(1/\delta))$ items with probability $1-\delta$.
\begin{lemma}
\label{online}
Let $A$ be the shortlist returned by Algorithm \ref{alg:main}, and $\delta$ is the parameter used to call Algorithm \ref{alg:SIIImax} in Algorithm \ref{alg:main}. Then, for given configuration $Y$, for any item $a$, we have $$Pr(a\in A|Y, a\in S_{1,\cdots, w}) \ge 1-\delta\ .$$
\end{lemma}
\begin{proof}
From Lemma~\ref{config} by conditioning on $Y$, the set $S_{1,\cdots, W}$ is determined. Now if $a\in S_{1,\dots, w}$, 
then for some slot $s_j$ in an $\alpha$ length subsequence $\tau$ of some window $w$, we must have 
$$a = \arg \max_{i\in s_j \cup R_{1, \ldots, w-1}} f(S_{1, \ldots,w-1} \cup \gamma(\tau) \cup \{i\}) - f(S_{1, \ldots, w-1} \cup \gamma(\tau)).$$ 
Let $w'$ be the first such window, $\tau', s_{j'}$ be the corresponding subsequence and slot. Then, it must be true that 
$$a = \arg \max_{i\in s_{j'}} f(S_{1, \ldots,w'-1} \cup \gamma(\tau') \cup \{i\}) - f(S_{1, \ldots, w'-1} \cup \gamma(\tau')).$$ 
(Note that the argmax in above is not defined on $R_{1,\cdots, w'-1}$).
The configuration $Y$ only determines the set of items in the items in slot $s_{j'}$, the items in $s_{j'}$ are still randomly ordered (refer to Lemma \ref{config}). Therefore, from Proposition~\ref{maxanalysis}, with probability $1-\delta$, $a$ will be added to the shortlist $A_{j'}(\tau')$ by Algorithm~\ref{alg:SIIImax}. Thus $a\in A \supseteq A_{j'}(\tau')$ with probability at least $1-\delta$.

\toRemove{If $a\in R_{1,\cdots, w-1}$ then $a$ has appeared in a window before $w$ for the first time, say $w'$. Since $a\in R_{1,\cdots, w-1}$ there is $\tau'$, such that
$a := \arg \max_{i\in s_{\ell}} f(S_{1, \ldots,w'-1} \cup \gamma(\tau') \cup \{i\}) - f(S_{1, \ldots, w'-1} \cup \gamma(\tau'))$ 
(Note that the argmax in above is not defined on $R_{1,\cdots, w'-1}$).
The configuration $Y$ only determines the set of items in the items in slot $s_{j'}$, the items in $s_{j'}$ are still randomly ordered (refer to Lemma \ref{..}).
The permutation of elements in $s_{\ell}$ defines whether or not the online Algorithm~\ref{alg:SIIImax} selects $a$ or not. 
Therefore, from Theorem~\ref{maxanalysis}, with probability $1-\delta$, $a$ will be added to the shortlist $A_{j'}(\tau')$ by Algorithm~\ref{alg:SIIImax}. Thus $a\in A \supseteq A_{j'}(\tau')$ with probability at least $1-\delta$.
Thus $a\in H_{1,\cdots, w'-1}$ with probability at least $1-\delta$.
Now If $a \notin R_{1,\cdots, w-1}$ and $a\in s_j$, then the permutation of elements in $s_{j}$ defines whether or not the online Algorithm~\ref{alg:SIIImax} selects $a$ or not.
Therefore again from theorem~\ref{maxanalysis}, 
$a\in H_{1,\cdots, w'-1}$ with probability at least  $1-\delta$.}
\end{proof}
\begin{prop}
\label{prop:online}
$$\Ex[f(A^*)] := \mathbb{E}[f(S_{1,\cdots, W} \cap A)] \ge (1-\frac{\epsilon}{2})\mathbb{E}[f(S_{1,\cdots, W})]$$
where $A^*:= S_{1,\cdots, W} \cap A$ is the size $k$ subset of shortlist $A$ returned by Algorithm \ref{alg:main}.
\end{prop}
\begin{proof}
From the previous lemma, given any configuration $Y$, we have that each item of $S_{1,\cdots, W}$ is in $A$ with probability at least $1-\delta$, where $\delta=\epsilon/2$ in Algorithm \ref{alg:main}.  
Therefore using Lemma~\ref{sample}, the expected value of $f(S_{1,\cdots, W}\cap A)$ 
is at least $(1-\delta)\mathbb{E}[F(S_{1,\cdots, W})]$.
\end{proof}
\paragraph{Proof of Theorem \ref{opttheorem}.} Now, we can show that Algorithm \ref{alg:main} provides the results claimed in Theorem \ref{opttheorem} for appropriate settings of $\alpha, \beta$ in terms of $\epsilon$. 
Specifically for $\delta'=\epsilon/4$, set $\alpha,\beta$ as smallest integers satisfying  \settingb. Then, using Proposition \ref{prop:first} and Proposition \ref{prop:online}, for $k\ge \alpha\beta$ we obtain:
$$\Ex[f(A^*)] \ge (1-\frac{\epsilon}{2})(1-\delta')^2 (1-1/e) \OPT \ge (1-\epsilon)(1-1/e) \OPT.$$
This implies a lower bound of $1-\epsilon - 1/e - \alpha\beta/k = 1-\epsilon-1/e - O(1/k)$ on the competitive ratio.

The $O(k)$ bound on the size of the shortlist was  demonstrated in Proposition \ref{prop:size}.


\section{Streaming (Proof of Theorem \ref{thm:streaming})}
\label{sec:streaming}
In this section, we show that Algorithm \ref{alg:main} can be implemented in a way that it uses a memory buffer of size at most $\eta(k)=O(k)$; and the number of objective function evaluations for each arriving item is $O(1+\frac{k^2}{n})$. This will allow us to obtain  Theorem \ref{thm:streaming} (restated below) as a corollary of Theorem \ref{opttheorem}.
\thmStreaming*

In the current description of Algorithm \ref{alg:main}, there are several steps in which the algorithm potentially needs to store $O(n)$ previously seen items in order to compute the relevant quantities. 
First, in Step \ref{li:subb}, in order to be able to compute $\gamma(\tau)$ for all less than $\alpha$ length subsequences $\tau$ of slots $s_1, \ldots, s_{j-1}$, the algorithm should have stored all the items that arrived in the slots $s_1, \ldots, s_{j-1}$. However, this memory requirement can be reduced by a small modification of the algorithm, so that at the end of iteration $j-1$, the algorithm has already computed  $\gamma(\tau)$ for all such $\tau$, and stored them to be used in iteration $j$. In fact, this can be implemented in a memory efficient manner, in the following way. For every subsequence $\tau$ of slots $s_1, \ldots, s_{j-1}$ of length $<\alpha$, consider prefix $\tau'=\tau\backslash s_{j-1}$. Assume $\gamma(\tau')$ is available from iteration $j-2$. 
If $\tau'=\tau$, then $\gamma(\tau)=\gamma(\tau')$. Otherwise, in Step 6 of iteration $j-1$, the algorithm must have considered the subsequence $\tau'$ while going through all subsequences of length less than $\alpha$ of slots $s_1, \ldots, s_{j-2}$. Now, modify the implementation of Step 6 so that  the algorithm also tracks the (true) maximum $M_{j-1}(\tau')$ of $a_0, a_1, \ldots, a_N$ for each $\tau'$. Then, $\gamma(\tau)$ can be obtained by extending $\gamma(\tau')$ by $M_{j-1}(\tau')$, i.e.,  $\gamma(\tau)=\{\gamma(\tau'), M_{j-1}(\tau')\}$. Thus, at the end of iteration $j-1$, $\gamma(\tau)$ would have been computed for all subsequences $\tau$ relevant for iteration $j$, and so on. In order to store these $\gamma(\tau)$ for every subsequence $\tau$ (of  at most $\alpha$ slots from $\alpha \beta$ slots), we require a memory buffer of size at most $\alpha^2{\alpha \beta \choose \alpha} = O(1)$.  

Secondly, across windows and slots, the algorithm keeps track of $R_w, S_w, w=1,\ldots, k/\alpha$ where $W=k/\alpha$. In the current description of Algorithm \ref{alg:main}, these sets are computed after seeing all the items in window $w$ in Step~\ref{li:Rw}. Thus, all the items arriving in that window would be needed to be stored in order to compute them, requiring $O(n)$ memory buffer. However, the alternate implementation discussed in the previous paragraph reduces this memory requirement to $O(k)$ as well. Using the above implementation, at the end of iteration $\alpha \beta$ for the last slot $s_{\alpha\beta}$ in window $w$, we would have computed and stored $\gamma(\tau)$ for all the subsequences  $\tau$ of length $\alpha$ of slots $s_1,\ldots, s_{\alpha\beta}$. 
$R_w$ is simply defined as union of all items in  $\gamma(\tau)$ over all such $\tau$ (refer to \eqref{eq:Rw}). And, $S_w = \gamma(\tau^*)$  for the best subsequence $\tau^*$ among these subsequences (refer to \eqref{eq:Sw}). 
Thus, computing $R_w$ and $S_w$ does not require any additional memory buffer. Storing $R_w$ and $S_w$ for all windows requires a buffer of size at most $\sum_w |R_w| + |S_w| = \frac{k}{\alpha} \times \alpha {\alpha \beta \choose \alpha}+ k = O(k).$
Therefore, the total buffer required to implement Algorithm \ref{alg:main} is of size $ O(k)$. 


Finally, let's bound the number of objective function evaluations for each arriving item. Each arriving item is processed in Step 6, where objective function is evaluated twice for each $\tau$ to compute the corresponding $a_i$. Since there are atmost ${\alpha \beta \choose \alpha}$ subsequences $\tau$ for which this quantity is computed, the total number of times 
this computation is performed is bounded by $2 {\alpha \beta \choose \alpha}=O(1)$. However, for each $\tau$, we also compute $a_0$ in the beginning of the slot. Computing $a_0$ for each $\tau$ involves taking max over all items in $R_{1,\ldots, w-1}$, and requires $2|R_{1,\ldots,w-1}|\le 2k {\alpha \beta \choose \alpha}$ evaluations of the objective function. Due to this computation, in the worst-case, the update time for an item can be $ 2k {\alpha \beta \choose \alpha}^2 + 2 {\alpha \beta \choose \alpha}= O(k)$. However, since $a_0$ is computed {\it once} in the beginning of the slot for each $\tau$, the  total update time over all items is bounded by $2k {\alpha \beta \choose \alpha}^2 \times k\beta + {\alpha \beta \choose \alpha} \times n = O(k^2+n)$. Therefore the amortized update time for each item is $O(1+\frac{k^2}{n})$.
\scomment{replaced by above:Also each subroutine call on a slot $s_j$ goes over all elements in $s_j$ and also $R_{1,\cdots, w-1}$ to find the maximum element. 
When we pass elements of $s_j\cup R_{1,\cdots, w-1}$ to the online subroutine~\ref{alg:SIIImax} one by one, for each element $e\in  s_j$ we can also pass $\frac{|R_{1,\cdots, w-1}|}{|s_j|}$ elements from $R_{1,\cdots, w-1}$. 
Therefore the amortized update time
would be $O(1+\frac{k^2}{n})$.
Also note that the worst case update time can be $O(k)$.}
This concludes the proof of Theorem \ref{thm:streaming}.



\section{Impossibility Result (Proof of Theorem \ref{hardness})}
\label{sec:hardness}
 

In this section we provide an upper bound showing the following:

\hardness*


In the following proof, for simplicity of notation, we prove the desired bound for submodular $(k+1)$-secretary problem. 
For any given $n, k$, we construct a set of instances of the  \nameOfProblemSLplus~such that any online algorithm that uses a shortlist of size $\eta(k+1)$ will have competitive ratio of at most $\frac{7}{8} + \frac{\eta(k+1)}{2n}$ on a randomly selected instance from this set. 

First, we define a monotone submodular function $f$ as follows. 
The ground set consists of $\frac{n}{2k}+n-1$ items. There are two types of items, $C$ and  $D$, with $L:=n/2k$ items of type $C$ and
$n-1$ items of type $D$. We define $f(\phi):=0$, $f(\{c\}):=k$ for $c\in C$, and $f(\{d\}):=1$ for all $d\in D$. 
Also there is a collection of $L$ {\it disjoint} sets $T_{\ell}=\{c^{\ell}, d_1^{\ell},\cdots, d^{\ell}_{k}\}$, $\ell=1,2, \ldots L$, such that $c^{\ell}\in C$ and $d_j^{\ell}\in D$. We define $f(T_{\ell}):=2k$ for all $\ell=1,\ldots, L$.
Now, let
$$g(t):=k+\frac{k}{2}+\cdots+ \frac{k}{2^{i-1}} + \frac{(t-ik)}{2^i},$$
where $i=\floor{t/k}$. It is easy to see that $g$ is a monotone submodular function. 

Now, define $f$ on the remaining subsets of the ground set as follows. For all $S$ with $|S|\ge 1$, 
\begin{itemize}
\item $|S\cap C| \ge 2 \implies f(S):=2k+1$
\item $|S\cap C|=0 \implies f(S):=1+g(|S|-1)$
\item $|S\cap C|=1$ $\implies S\cap C=\{c^\ell\}$ for some $\ell\in [L]$   $\implies$
$$f(S):= \min\{2k+1, k+\frac{1}{2} g(|S|-1)+\frac{k'}{2^{i+1}}\},$$ 
where $k'=|S\cap\{d^{\ell}_1,\cdots, d_k^{\ell}\} |$, $i=\lfloor(|S|-1)/k\rfloor$.
\end{itemize}

Observe that since $g(k)=k$, for any subset $S$ of size at most $k+1$, we have $f(S)\le k+\frac{k}{2} + \frac{k}{2} = 2k$.

\begin{lemma}
$f$ is a monotone submodular function.
\end{lemma}
\begin{proof}
We have to show that for any item $x$ and subsets $S\subseteq T$, $\Delta_f(x|S) \ge \Delta_f(x|T)$.
We consider the following cases:
\begin{itemize}
\item if $|T\cap C| \ge 2 \implies \Delta_f(x|T)=0$, so it is trivial.
\item if $|T\cap C| =0 \implies |S\cap C|=0 \implies \Delta_f(x|S) \ge \Delta_f(x|T)$ because of submodularity of $g$.
\item if $|T\cap C|=1 \implies |S\cap C|\le 1$
\begin{itemize}
\item if $|S\cap C|=1 $ then $S\cap C = T\cap C = \{c^{\ell}\}$ for some $\ell$:
\begin{itemize}
\item $x\in \{d_1^{\ell},\cdots, d_k^{\ell}\} \implies \Delta_f(x|S)=1/2^{i+1}+1/2^{i+1}$ for $i=\lfloor (|S|-1)/k\rfloor$, and $\Delta_f(x|T)=1/2^j+1/2^j$ for some $j=\lfloor (|T|-1)/k\rfloor$ and $j\ge i+1$.
\item $x\notin \{d_1^{\ell},\cdots, d_k^{\ell}\} \implies \Delta_f(x|S)=1/2^{i+1}$ for $i=\lfloor (|S|-1)/k\rfloor$ and $\Delta_f(x|T)=1/2^j$ for some $j\ge i+1$.
\end{itemize} 

\item if $|S\cap C|=0 \implies \Delta_f(x|T) \le 1/2^{j+1}+1/2^{j+1}$ for $j=\lfloor (|T|-1)/k\rfloor$ and $\Delta_f(x|S)= 1/2^i$ for some $i\le j$. 
\end{itemize}
Thus $\Delta_f(x|S) \ge \Delta_f(x|T)$.
\end{itemize} 
Monotonicity follows trivially from the definition of $f$.
\end{proof}

\toRemove{
Also for any subset of $S$ with at least two elements of $C$, $f(S)=2k$. 
For subset $S$ with no element in $C$, if $i=\floor*{\frac{|S|}{k}}$, 
$f(S)$ is defined as 
\[
f(S)= k+k/2+ \dots +(|S|-ik)/2^{i}.
\]
For subset $S$ with exactly one element $c\in C$ such that $c\neq c_1$, if if $i=\floor*{\frac{|S|-1}{k}}$, 
$f(S)$ is defined as 
\[
f(S)= k+ k/2+ \dots +(|S|-1-ik)/2^{i+1}.
\]
Thus if $|S|=k+1$ and it contains exactly one element of $c\in C$ such that $c\neq c_1$ then $f(S)=k+k/2$.
Now suppose $S$  contains $c'$. 
For $S=T$, $f(S)=2k$.  Suppose $S$ contains $k'<k$ elements of $\{b_1, \dots, b_k\}$  and $\ell$ other elements of $D$. 
Let's $t=k'+\ell$, and $i=\floor*{\frac{t}{k}}$, then
\[
f(S)= \min\{2k, k'/2+(k+ k/2+ \dots +(t-ik)/2^{i+1})\}
\]
}


Now, denote $D^{\ell} := T^\ell \cap D = \{d_1^{\ell},\cdots, d_k^{\ell}\}$ for $\ell=1,2, \ldots, L$. Also, let 
$D'=D\setminus (\bigcup_{\ell=1}^L D^{\ell})$. 
Now define $L$ input instances $\{I_{\ell}\}_{\ell =1,\ldots, L}$, each of size $n$, as follows. For
any arbitrary subset $\tilde{D}\subseteq D'$ of size $n-Lk-1$, 
define $I_{\ell}=\bigcup_{i=1,\ldots, L}  D^{i} \cup \tilde{D} \cup \{c^{\ell}\}$, for  $\ell=1,\ldots, L$.
Thus, for instance $I_{\ell}$, the the optimal $k+1$ subset is $T^\ell$ with value $f(T^{\ell})=2k$. 

Now consider any algorithm for the submodular \nameSecSL\ and cardinality constraint $k+1$. We denote by $Alg$ the set of $\eta(k+1)$ items selected by the algorithm as part of the shortlist. 
Let $\bar I$ denote an instance chosen uniformly at random from $I_\ell, \ell=1,\ldots, L$.  
Let $\pi$ denote a random ordering of $n$ items in $\bar I$. 
We denote by random variable $(\bar{I},\pi)$ the randomly ordered input instance to the algorithm. Also we denote by $\bar{T}, \bar{D}$ and $\bar c$, the corresponding $T^{\ell}$, $D^{\ell}$ and $c^\ell$.

Now we claim 


\begin{lemma}
$\mathbb{E}_{(\bar{I},\pi)}[|Alg\cap \bar{D}|] \le k/2+\eta(k+1)/L$.
\end{lemma}
\begin{proof}
Suppose $(e_1,\cdots, e_n)$ indicates the ordered input according to random ordering $\pi$ on $\bar{I}$. 
Now let $t$ be the random variable indicating the index of $c^{\ell}$ in $(e_1,\cdots, e_n)$, i.e., $e_t=c^{\ell}$. 
Then, due to random ordering, and random choice of $\bar I$ from $I_1, \ldots, I_\ell$, we have
$$\mathbb{E}_{(\bar I, \pi)}[|Alg\cap\{e_1,\cdots,e_{t-1}\}\cap D^1|]= \cdots = \mathbb{E}[|Alg\cap\{e_1,\cdots, e_{t-1}\}\cap D^{L}|]\ .$$
Also, since $D^\ell, \ell=1,\ldots, L$ are disjoint,
$$\sum_{\ell=1}^L \mathbb{E}[|Alg\cap\{e_1,\cdots,e_{t-1}\}\cap D^{\ell}|] \le \eta(k+1)\ .$$
Since $\bar D=D^\ell$ with probability $1/L$,  we have
$$H:=\mathbb{E}[|Alg\cap\{e_1,\cdots, e_{t-1}\}\cap \bar{D}|]=  \frac{1}{L} \sum_{\ell=1}^L \mathbb{E}[|Alg\cap\{e_1,\cdots,e_{t-1}\}\cap D^{\ell}|] \le \frac{1}{L} \eta(k+1)\ . $$ 
Now define
$G:=\mathbb{E}[|Alg\cap\{e_t,\cdots, e_{n}\}\cap \bar{D}|]$.
We have
$$G\le \mathbb{E} [|\bar{D}\cap\{e_t,\cdots, e_n\}|]  \le k/2\ .$$
Thus $$\mathbb{E}[|Alg\cap \bar{D}|] \le G+H\le k/2+\eta(k+1)/L\ .$$
\end{proof}
Now on input $\bar{I}$, if the algorithm doesn't select $\bar c$ as part of shortlist $Alg$, then by definition of $f$ for sets that do not contain any item of type $C$, we have 
$$f(A^*):=\max_{S \subseteq Alg:|S|\le k+1} f(S) \le 1+ g(k) = k+ \frac{k}{2}$$ 

Otherwise, if algorithm selects $\bar c$ then by definition of $f$ 
$$f(A^*):=\max_{S \subseteq Alg:|S|\le k+1} f(S) \le \max_{S \subseteq Alg \backslash (\bar D \cup \{\bar c\}):|S|\le k-|Alg \cap \bar D|} f(S \cup \bar D \cup \{\bar c\}) 
= k+ \frac{k}{2} +\frac{1}{2} |Alg\cap \bar{D}|$$ 
therefore $$\mathbb{E}[f(A^*)] \le k+\frac k 2+ \frac k 4+ \frac{\eta(k+1)}{2L} = \frac{7k}{4} + \frac{k\eta(k+1)}{n} \ .$$

Since the optimal is equal to $\Ex[f(\bar T)]=2k$, the competitive ratio is upper bounded by
$$ 
\frac{7}{8} + \frac{\eta(k+1)}{2n}$$

This proves competitive ratio upper bound of $\frac{7}{8} + o(1)$ when $\eta(k+1)=o(n)$, to complete the proof of Theorem \ref{hardness}.


 


\bibliographystyle{plainnat}
\bibliography{mybib}

\section{Appendix}
\subsection{Some useful properties of $(\alpha,  \beta)$  windows}
\label{app:windows}
Lemma \ref{lem:indep} is a corollary of the following lemma.
\begin{lemma}\label{eqprob}
For each $y\in [k\beta]^n$, $Pr\{Y=y\}=(\frac{1}{k\beta})^n$.
\end{lemma}
\begin{proof}
Consider pair $(\pi, \psi)$, where $\pi:I\to [n]$  defines the random order on $I$. 
Throw $n$ balls uniformly into $k\beta$ bins. Let $\psi_j$ be the  bin that $j$-th ball goes into. 
Note that $\psi$ and $\pi$ are independent. Now consider
\[
(\frac{1}{k\beta} s_1 + \cdots + \frac{1}{k\beta} s_{k\beta})^n = 
\left(\frac{1}{k\beta}\right)^n \sum_{t_1,\cdots, t_{k\beta}} Q_{t_1,\cdots, t_{k\beta}}  {s_1}^{t_1} \cdots  {s_{k\beta}}^{t_{k\beta}}\ .
\]
For a given $y\in  [k\beta]^I$,  suppose $t_i$ is the number of elements in slot $s_i$.
Then from above expansion the probability that $\psi$ divides input into slots of size $t_1,\cdots, t_{k\beta}$ is
\[
\left(\frac{1}{k\beta}\right)^n Q_{t_1,\cdots, t_{k\beta}}  =  \left(\frac{1}{k\beta}\right)^n {n \choose t_1, t_2, \cdots, t_{k\beta}}\ .
\]
Now for such a $\psi$, the probability that permutations $\pi$ satisfy $Y=y$ is 
\[
\frac{t_1! \cdots t_{k\beta}! }{n!}\ .
\]
Thus the probability that $Y=y$ is 
\[
\left(\frac{1}{k\beta}\right)^n {n \choose t_1, t_2, \cdots, t_{k\beta}} \frac{t_1! \cdots t_{k\beta}! }{n!} = \left(\frac{1}{k\beta}\right)^n\ .
\]

\end{proof}

\subsection{m-submodular functions}
\label{app:msubm}
\begin{definition}
We call a function $f:2^{A}\rightarrow \mathbb{R}, $ m-submodular if it is submodular and there exists a submodular function $F$ such that:
\[ 
f(S)= \max_{T\subseteq S, |T|\le m} F(T)\ .
\]
\end{definition}
Note that  maximum node weighted bipartite matching and  maximum edge weighted bipartite matching defined on $G=(X\times Y)$ with $|Y|=m$ are m-submodular. 
(the assignments will be done at the end of algorithm after all the selections are made )

\begin{remark}
$f(S)=\max_{a\in S} a$ is a $1$-submodular function.
\end{remark}

Now consider the following simple greedy algorithm:

\begin{algorithm*}[ht]
  \caption{~\bf{Select-If-it-Improves}($f,I , u$)}
  \label{alg:SIII} 
\begin{algorithmic}[1]
\State $R\leftarrow \emptyset$
\For {i=0 to n}
\If {$f(R\cup \{a_i\}) > f(R) $} 
\State $R\leftarrow R\cup \{a_i\}$
\EndIf
\EndFor
\State return $S\leftarrow R\setminus \{a_1,\cdots, a_u\}$
\end{algorithmic}
\end{algorithm*}

\scomment{the remark below is incorrect. We have changed Algorithm \ref{alg:SIIImax} to start at $u$, and also end once the shortlist is of size $L$.}


\scomment{please use "theorem" very sparingly. Theorem should be used only for the main results of the paper. We should not have any theorems beyond the 4 results in "Our results" section. I have changed your "theorems" to "lemma"}
\begin{lemma}
Suppose $R$ is the set of elements selected in the above algorithm on the input $I=\{a_1,\cdots, a_n\}$ then $f(R) = f(I)$.
\end{lemma}
\begin{proof}
Suppose $R_i$ is the subset selected at iteration $i$. Since $f$ is submodular, if $f(R_i\cup \{a_i\}) \le f(R_i)$ 
then $f(R\cup\{a_i\}) \le f(R)$. Therefore every $e\in I\setminus R$ has marginal value 0 with respect to $R$, i.e., $f(R)=f(I)$.
\end{proof}


\begin{lemma}
$E[|S|] = m\ln(n/u)$.
\end{lemma}
\begin{proof}
Suppose $f(R_i) = F(T)$, where $|T|=m$. If $a_i \notin T$ then it is not selected. 
Because if $a_i\notin T$ and is selected then it should have positive $f$ marginal value, which means $f(R_i) = f(R_{i-1} \cup \{a_i\} ) > f(R_{i-1}) =  F(T) $, it is a contradiction. 
Thus only elements in $T$ will be selected at position $i$.  

If you consider all permutations of $R_i$, an element will be selected at position $i$ if it is subset of $T$, the probability is $|T|/i= m/i$. Therefore the total expected number of selections $\mathbb{E}[|R|]$, will be at most $\sum_{i=1}^{n} \frac{m}{i}  = m \ln n $. Similarly 
$\mathbb{E}[|S|] \le \sum_{i=u}^{n} \frac{m}{i}=m\ln(n/u)$.
\end{proof}

\toRemove{
\begin{lemma}
$E[|S|] = m\log n$
\end{lemma}
\begin{proof}
Suppose $f(S_i) = F(T)$, where $|T|=m$. If $a_i \notin T$ then it is not selected. 
Because if $a_i\notin T$ and is selected then it should have positive $f$ marginal value, which means $f(S_i) = f(S_{i-1} \cup \{a_i\} ) > f(S_{i-1}) =  F(T) $, it is a contradiction. 
Thus only elements in $T$ will be selected at position $i$.  

If you consider all permutations of $S_i$, an element will be selected at position $i$ if it is subset of $T$, the probability is $|T|/i= m/i$. Therefore the total expected number of selections will be at most $\sum_{i=1}^{n} \frac{m}{i}  = m \ln n $
\end{proof}
}

In the rest we will make the following assumption:\\
\textbf{Assumption. } 
There is a unique optimal solution OPT.


\begin{lemma} \label{hprob}
Algorithm~\ref{alg:SIII}, with parameter $u=n\epsilon$,  selects a set $S$ with 
$$|S|<m\ln (1/\epsilon)+\ln(1/\delta)+\sqrt{\ln^2{1/\delta}+2m\ln(1/\delta)\ln(1/\epsilon)}$$ 
and $E[f(S)]=(1-\epsilon-\delta)OPT$.
\end{lemma}
\begin{proof}
We use Freedman's inequality.
 If $\{a_1,\cdots, a_i\}$ has a unique maximum subset of size $m$, define $Y_i$ to be a random variable indicating whether the algorithm has selected $a_i$  or not, where $Y_i=1-\frac{m}{i}$ if $a_i$ is selected and $Y_i=-\frac{m}{i}$ otherwise. 
 If it has not unique solution define $Y_i=0$. ($a_i$ will not be selected)
 Also define $\mathcal{f}_i=\{Y_{n},Y_{n-1}, \cdots, Y_{n-i+1}\}$.

Let $X_i=\sum_{j=n-i+1}^{n} Y_j$,
then $\{X_i\}$ is a martingle, 
because $E[X_{i+1}|\mathcal{f}_{i}] = X_i+E[Y_{n-i}|\mathcal{f}_i]$.
 If $\{a_1,\cdots, a_i\}$ has a unique maximum subset of size $m$, $E[Y_{n-i}|\mathcal{f}_i]=(m/i)(1-m/i)+(1-m/i)(-m/i)=0$,
 otherwise $E[Y_{n-i}|\mathcal{f}_i]=0$. So in both cases $E[X_{i+1}|\mathcal{f}_{i}] =X_i$.
As in the Freedman's inequality, let $L=\sum_{i=n\epsilon}^{n} Var(Y_i| f_{i-1})$. 
\begin{align*}
L  =  \sum_{i=n\epsilon}^{n} \frac{m}{i}  (1-\frac{m}{i})^2 + (1-\frac{m}{i}) (\frac{m}{i})^2 
< \sum_{i=n\epsilon}^{n} \frac{m}{i} =m\ln (1/\epsilon) \ .
\end{align*}
Therefore,
\[
Pr(X_{n-n\epsilon}\ge \alpha \text{ and }  L\le m\ln (1/\epsilon) ) \le exp(-\frac{\alpha^2}{2m\ln (1/\epsilon)+ 2\alpha })   < \delta \ .
\]
Thus we get $\alpha > \ln(1/\delta)+\sqrt{\ln^2{1/\delta}+2m\ln(1/\delta)\ln(1/\epsilon)}$.
Also $|S| = X_{n-n\epsilon} + m\ln(1/\epsilon)$. Therefore 
\[
Pr(|S| \ge  m\ln (1/\epsilon)+\ln(1/\delta)+\sqrt{\ln^2{1/\delta}+2m\ln(1/\delta)\ln(1/\epsilon)}  )  \le \delta\ .
\]
So with probability $(1-\delta)$, $|S| \le m\ln (1/\epsilon)+\ln(1/\delta)+\sqrt{\ln^2{1/\delta}+2m\ln(1/\delta)\ln(1/\epsilon)} $.  
Since $F$ is submodular, $E[F(OPT\cap \{a_{n\epsilon},\cdots, a_n \} )] = (1-\epsilon)OPT$.
Therefore $E[f(S)] \ge (1-\epsilon)OPT-\delta OPT$. 
\end{proof}
\scomment{Considering that Remark 1 is incorrect,  proof of Proposition \ref{maxanalysis} does not directly follow from above. We need to explicitly add the proof of Proposition \ref{maxanalysis}. Or else you need modify Algorithm 3 so that Algorithm~\ref{alg:SIIImax} is a special case.}



\maxanalysis*
\begin{proof}
Set $u=n\delta/2$ and $\epsilon=\delta/2$, and $f(T):=\max_{a\in T} a$.
 The set $S$ returned by the Algorithm~\ref{alg:SIII} is the same as the set $A$ selected by Algorithm~\ref{alg:SIIImax},  when $|S|<L$.
From lemma~\ref{hprob} with probability $(1-\delta)$, $|S|<(3+\sqrt{2})\ln(2/\delta)<L$. 
Also $f(S)=A^*$.
Therefore w.p. $(1-\delta)$ Algorithm~\ref{alg:SIIImax} returns $A^*$.
\end{proof}

\end{document}